\documentclass[journal]{IEEEtran}

\usepackage{graphicx}
\usepackage{amsfonts}
\usepackage{amsmath}
\usepackage{amssymb}

\newtheorem{thm}{Theorem}%[section]

\newtheorem{lem}[thm]{Lemma}
\newtheorem{proof}[thm]{proof}

\newtheorem{defn}[thm]{Definition}
\newtheorem{rem}[thm]{Remark}
\newtheorem{exam}[thm]{Example}

\numberwithin{equation}{section}

\begin{document}

\title{Optimal Locally Repairable Linear Codes}

\author{Wentu Song,~Son Hoang Dau,~Chau Yuen~and~Tiffany Jing Li
\thanks{W. Song, S. H. Dau and C. Yuen are with Singapore University of
Technology and Design, Singapore
       (e-mails: \{wentu\_song, sonhoang\_dau, yuenchau\}@sutd.edu.sg).}
\thanks{T. J. Li is with the Department of Electrical and Computer Engineering,
Lehigh University, Bethlehem, PA 18015, USA (e-mail:
jingli@ece.lehigh.edu).}
%\thanks{Manuscript received April 19, 2005; revised December 27, 2012.}
}

\maketitle

\begin{abstract}
%One major challenge in designing erasure codes for distributed storage systems is to guarantee an efficient repair mechanism with modest bandwidth and computation requirements when recovering data from failed storage devices.
%Linear erasure codes that allow an erased symbol to be locally repaired by accessing a small number of other symbols are desirable for distributed data storage networks.

Linear erasure codes with local repairability are desirable for
distributed data storage systems. An $[n,k,d]$ code having
all-symbol $(r,\delta)$-locality, denoted as $(r,\delta)_a$, is
considered optimal if it also meets the minimum Hamming distance
bound. The existing results on the existence and the  construction
of optimal $(r, \delta)_a$ codes are limited to only the special
case of $\delta=2$, and to only two small regions within this
special case, namely, $m=0$ or $m\ge (v+\delta-1)>(\delta-1)$,
where $m=n \text{\ mod\ }(r+\delta-1)$ and $v=k~ \text{mod}\ r$.
This paper investigates the existence conditions and presents
deterministic constructive algorithms for optimal $(r,\delta)_a$
codes with general $r$ and $\delta$. First, a structure theorem is
derived for general optimal $(r,\delta)_a$ codes which helps
illuminate some of their structure properties. Next, the entire
problem space with arbitrary $n$, $k$, $r$ and $\delta$ is divided
into eight different cases (regions) with regard to the specific
relations of these parameters. For  two cases, it is rigorously
proved that no optimal $(r,\delta)_a$ could exist. For four other
cases the  optimal $(r,\delta)_a$ codes are shown to exist,
deterministic constructions are proposed and the lower bound on
the required field size for these algorithms to work is provided.
Our new constructive algorithms not only cover more  cases, but
for the same cases where previous algorithms exist, the new
constructions require a considerably smaller field, which
translates to potentially lower computational complexity. Our
findings substantially enriches the knowledge  on  $(r,\delta)_a$
codes, leaving only two cases in which the existence of optimal
codes are yet to be determined.

%    \item[(ii)] we show that optimal codes do not exit for certain sets of    coding parameters,
%    \item[(iii)] we provide deterministic constructions for optimal codes    for four sets of coding parameters, which greatly extends the existence results of optimal codes; consequently, the existence of optimal codes remains unknown for only two sets of parameters.
%\end{itemize}
% Additionally, our constructions require much smaller field sizes
%than those in the previous work.
\end{abstract}

%\begin{IEEEkeywords}
%Network coding, region decomposition, k-pair problem.
%\end{IEEEkeywords}

\IEEEpeerreviewmaketitle

\section{Introduction}

The sheer volume of today's digital data has made {\it distributed
storage systems} $($DSS$)$ not only massive in scale but also
critical in importance. Every day, people knowingly or unknowingly
connect to various private and public distributed storage systems,
include large data centers (such as the Google data centers and
Amazon Clouds) and peer-to-peer storage systems (such as
OceanStore \cite{Rhea}, Total Recall \cite{Bhagwan}, and DHash++
\cite{Dabek}). In a distributed storage system, a data file is
stored at a distributed collection of storage devices/nodes in a
network. Since any storage device is individually
 unreliable and subject to failure (i.e. erasure), redundancy must be introduced to provide the much-needed system-level protection against data loss due to device/node failure.

The simplest form of redundancy is {\it replication}. By storing $c$
identical copies of a file at $c$ distributed nodes, one copy per node, a $c$-replication system can guarantee the data availability as long as no more than $(c\!-\!1)$ nodes fail. Such systems are very easy to implement, but extremely inefficient in storage space utilization, incurring tremendous waste in devices and equipment, building space, and cost for powering and cooling. More sophisticated systems employing {\it erasure coding} \cite{Weather02} can expect to considerably improve the storage efficiency. Consider a file that is divided into $k$ equal-size fragments. A judiciously-designed $[n,k]$ erasure (systematic) code can be employed to encode the $k$ data fragments (terms {\it systematic symbols} in the coding jargon) into $n$ fragments (termed {\it coded symbols}) stored in $n$ different nodes. If the $[n,k,d]$ code reaches the Singleton bound such that the minimum Hamming distance satisfies $d=n-k+1$, then the code is {\it maximum distance separable} (MDS) and offers redundancy-reliability optimality.  With an $[n,k]$ MDS erasure code, the original file can be recovered from any set of $k$ encoded fragments, regardless of whether they are systematic or parity. In other words, the system can tolerate up to $(n-k)$ concurrent device/node failure without jeopardizing the data availability.

Despite the huge potentials of MDS erasure codes, however, practical application of these codes in massive storage networks have been difficult. Not only are simple (i.e. requires very little computational complexity) MDS codes very difficult to construct, but data repair would in general require the access of $k$ other encoded fragments \cite{Rodrigues05}, causing considerable input/output (I/O) bandwidth that would pose huge challenges to a typical storage network.

%. This performance is optimal in terms of the redundancy-reliability tradeoff but induces relatively huge overheads in node repair, since by the MDS
%property, the repair of a single missing fragment will require
% $k$ other encoded fragments \cite{Rodrigues05}.
%%%%%%%%%%%%%%%%%%%%%%%%%%%%%%%%%%%%%%%%%%%
%\renewcommand\figurename{Fig}
\begin{figure}[htbp]
\begin{center}
\includegraphics[scale=0.73]{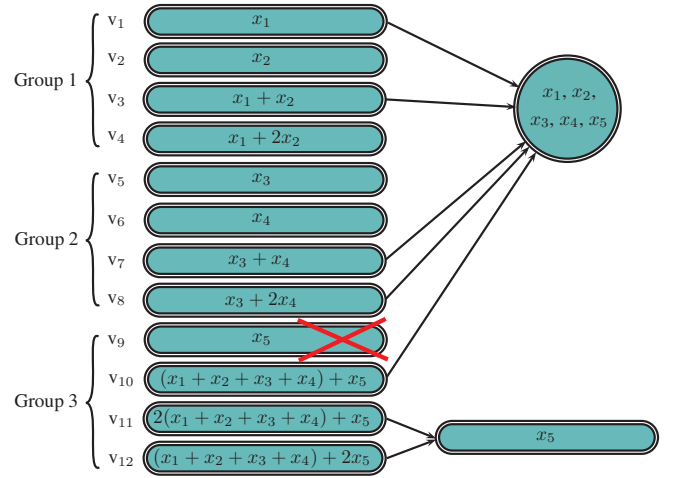}
\end{center}
\vspace{-5pt}
\caption{An example of how a locally repairable linear code is
used to construct a distributed storage system: a file $\mathcal
F$ is first split into five equal packets $\{x_1,\cdots,x_5\}$ and then
is encoded into $12$ packets, using a $(2,3)_a$ linear code. These
$12$ encoded packets are stored at $12$ nodes
$\{\text{v}_{1},\cdots,\text{v}_{12}\}$, which are divided into
three groups $\{v_1,v_2,v_3,v_4\}$, $\{v_5,v_6,v_7,v_8\}$ and
$\{v_9,v_{10},v_{11},v_{12}\}$. Each group can perform local repair of up to
two node-failures. For example, if Node $\text{v}_9$ fails, it
can be repaired by any two packets among
$\text{v}_{10},\text{v}_{11}$ and $\text{v}_{12}$. Moreover, the
entire file $\mathcal F$ can be recovered by five packets from any five nodes
$\text{v}_{i_1}, \cdots,\text{v}_{i_5}$ which intersect each
group with at most two packets. For example, $\mathcal F$ can be
recovered from five packets stored at
$\text{v}_1,\text{v}_3,\text{v}_7,\text{v}_8$ and $\text{v}_{10}$.
}\label{fig-dss}
\vspace{-17pt}
\end{figure}
%%%%%%%%%%%%%%%%%%%%%%%%%%%%%%%%%%%%%%%%%%%%%%

\vskip 10pt
Motivated by the desire to reduce repair cost in
 the design of  erasure codes for distributed storage systems, Gopalan
\emph{et al.}~\cite{Gopalan12} introduced the interesting notion
of {\it symbol locality} in linear codes. The $i$th coded symbol
of an $[n,k]$ linear code ${\mathcal C}$ is said to have locality
$r~(1\leq r\leq k)$ if it can be recovered by accessing at most
$r$ other symbols in $\mathcal C$. The concept was further
generalized to $(r,\delta)$ locality by Prakash \emph{et al.}
\cite{Prakash12}, to address the situation of multiple device
failures.

According to \cite{Prakash12}, the $i$th code symbol $c_i, 1\leq
i\leq n$, in an $[n,k]$ linear code $\mathcal C$ is said to have
locality $(r,\delta)$ if there exists an index set
$S_i\subseteq[n]$ containing $i$ such that $|S_i|-\delta+1\leq r$
and each symbol $c_j, j\in S_i$, can be reconstructed by any
$|S_i|-\delta+1$ symbols in $\{c_\ell;\ell\in S_i\text{~and~}
\ell\neq j\}$, where $\delta\geq 2$ is an integer. Thus, when
$\delta = 2$, the notion of locality in \cite{Prakash12} reduces
to the notion of locality in \cite{Gopalan12}. Two cases of
$(r,\delta)$ codes are introduced in the literature: An
$(r,\delta)_i$ code is a systematic linear code whose {\it
information symbols} all have locality $(r,\delta)$; and an
$(r,\delta)_a$ code is a linear code all of whose {\it symbols}
have locality $(r,\delta)$. Hence, an $(r,\delta)_a$ code is also
referred to as having \emph{all-symbol locality} $(r,\delta)$, and
an $(r,\delta)_i$ code is also referred to as having
\emph{information locality} $(r,\delta)$. A symbol with
$(r,\delta)$ locality -- given that at the most $(\delta-1)$
symbols are erased -- can be deduced by reading at most $r$ other
unerased symbols.

Clearly, codes with a low symbol locality, such as $r<k$, impose a
low I/O bandwidth and repair cost in a distributed storage system.
In a DSS system, one can use ``group'' to describe storage nodes
situated in the same physical location which enjoy a higher
communication bandwidth and a shorter communication distance than
storage nodes belonging to different groups. In the case of node
failure, a \emph{locally repairable code} makes it possible to
efficiently recover data stored in the failed node by downloading
information from nodes in the same group (or in a minimal number
of other groups).  Fig. \ref{fig-dss} provides a simple example of
how an $(r,\delta)_a$ code is used to construct a distributed
storage system. In this example, $\mathcal C$ is a $(2,3)_a$
linear code of length $12$ and dimension $5$. Note that a failed
node can be reconstructed by accessing only two other existing
nodes, while it takes five existing nodes to repair a failed node
if a $[12,5]$ MDS code is used.

%%%%%%%%%%%%%%%%%%%%%%%%%%%%%%%%%%%%%%%%%%%
\renewcommand\figurename{Fig}
\begin{figure*}[htbp]
\begin{center}
\includegraphics[width=18.2cm]{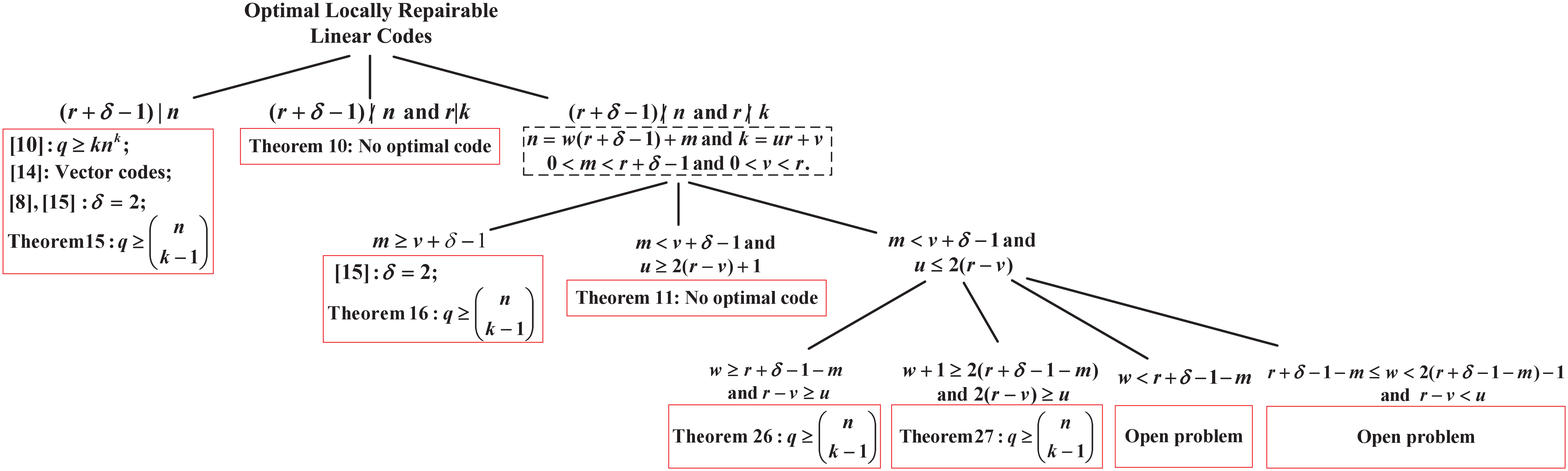}
\end{center}
\caption{Summary of existence of optimal $(r,\delta)_a$ linear
codes.}\label{sumy}
\end{figure*}
%%%%%%%%%%%%%%%%%%%%%%%%%%%%%%%%%%%%%%%%%%%%%%

\subsection{Related Work}
Locality was identified as a repair cost metric for distributed
storage systems independently by Oggier \emph{et al.}
\cite{Oggier11}, Gopalan \emph{et al.} \cite{Gopalan12} and
PaPailiopoulos \emph{et al.} \cite{Papail121} using different
terms. In \cite{Gopalan12}, Gopalan \emph{et al.} introduced the
concept of symbol locality of linear codes and established a tight
bound for the redundancy in terms of the message length, the
distance, and the locality of information coordinates. A
generalized concept, i.e., $(r,\delta)$ locality, was addressed by
Prakash \emph{et al.} \cite{Prakash12}. It was proved in
\cite{Prakash12} that the minimum distance $d$ of an
$(r,\delta)_i$ linear code $\mathcal C$ is upper bounded by
\begin{align}
d\leq n-k+1-\left(\left \lceil\frac{k}{r}\right \rceil-1\right )
(\delta-1)\label{eqn:1}
\end{align}
where $n$ and $k$ are the length and dimension of $\mathcal C$
respectively. It was also proved that a class of codes known as
pyramid codes \cite{Huang07} achieve this bound. Since an
$(r,\delta)_a$ code is also an $(r,\delta)_i$ code, (\ref{eqn:1})
also presents an upper bound for the minimum distance of
$(r,\delta)_a$ codes.

Locality of general codes (linear or nonlinear) and bounds on the
minimum distance for a given locality were presented in parallel
and subsequent works \cite{Papail122,Rawat12}. An $(r,\delta)_a$
code (systematic or not) is also termed a \emph{locally repairable
code (LRC)}, and  $(r,\delta)_a$ codes that achieve the minimum
distance bound are called \emph{optimal}.

It was proved in \cite{Prakash12} that there exists optimal
locally repairable linear codes when $(r+\delta-1)|n$ and
$q>kn^k$. Under the condition that $(r+\delta-1)|n$, a
construction method of optimal locally repairable vector codes was
proposed in \cite{Rawat12}, where maximal rank distance (MRD)
codes were used along with MDS array codes. For the special case
of $\delta=2$, Tamo \emph{et al.} \cite{Tamo13} proposed an
explicit construction of optimal LRCs when
$$(r+1)|n$$ or $$n~\text{mod}~(r+1)-1\geq k~\text{mod}~r>0.
\footnote{Note that this condition is equivalent to the condition
that $m\geq v+1$, where $n=w(r+1)+m$ and $k=u(r+1)+v$ satisfying
$0<m<r+1$ and $0<v<r$.}$$ Except for the special case that
$n~\text{mod}~(r+1)-1\geq k~\text{mod}~r>0$, no results are known
about whether there exists optimal $(r,\delta)_a$ code when
$(r+\delta-1)\nmid n$.

Up to now, designing LRCs with optimal distance remains an intriguing open
problem  for most coding
parameters $n, k, r$ and $\delta$.
Since large fields involve rather complicated and expensive computation,
a related interesting open problem asks how to limit the design
(of optimal LRCs) over relatively smaller fields.

\subsection{Main Results}
In this paper, we investigate the structure properties and the
construction of optimal $(r,\delta)_a$ linear codes of length $n$
and dimension $k$. A simple property of optimal $(r,\delta)_a$
linear codes is proved in Lemma \ref{low-bound}, which shows that
$\frac{n}{r+k-1}\geq\frac{k}{r}$ for any optimal $(r,\delta)_a$
linear code. Hence  we impose this condition of
$\frac{n}{r+k-1}\geq\frac{k}{r}$ throughout our discussion of
optimal $(r,\delta)_a$ codes.

The main results of this
paper include:

(\romannumeral1) We prove a structure theorem for the optimal
$(r,\delta)_a$ linear codes for $r|k$. This structure theorem
indicates that it is possible for optimal $(r,\delta)_a$ linear
codes, a sub-class of optimal $(r,\delta)_i$ linear code, to have
a simpler structure than otherwise.

(\romannumeral2) We prove that there exist no
optimal $(r,\delta)_a$ linear codes for
\begin{align}
(r+\delta-1)\nmid n~\text{and}~r|k \label{eqn:2}
\end{align}
or
\begin{align}
m<v+\delta-1~\text{and}~u\geq 2(r-v)+1 \label{eqn:3}
\end{align}
where $n=w(r+\delta-1)+m$ and $k=ur+v$ such that $0<v<r$ and
$0<m<r+\delta-1$ (Theorems \ref{non-exst} and \ref{non-exst-1}).

(\romannumeral3) We propose a deterministic algorithm for
constructing optimal $(r,\delta)_a$ linear codes over any field of
size $q\geq\binom{n}{k-1}$ when
\begin{align}
(r+\delta-1)|n \label{eqn:4}
\end{align}
or
\begin{align}
m\geq v+\delta-1 \label{eqn:5}
\end{align}
where $n=w(r+\delta-1)+m$ and $k=ur+v$ such that $0<v<r$ and
$0<m<r+\delta-1$ (Theorem \ref{opt-ext-1} and \ref{opt-ext-2}).

(\romannumeral4) We propose another deterministic algorithm for
constructing optimal $(r,\delta)_a$ linear codes over any field of
size $q\geq\binom{n}{k-1}$ when
\begin{align}
w\geq r+\delta-1-m \text{\ and\ } \text{min}\{r-v,w\}\geq u
\label{eqn:6}
\end{align}
or
\begin{align}
w+1\geq 2(r+\delta-1-m)~\text{and}~\text{min}\{2(r-v),w\}\geq u
\label{eqn:7}
\end{align}
 where $n=w(r+\delta-1)+m$ and $k=ur+v$ such that
$0<v<r$ and $0<m<r+\delta-1$ (Theorem \ref{opt-ext-3} and
\ref{opt-ext-4}).

A summary of our results is given in Fig \ref{sumy}. Note that if
none of the conditions in (\ref{eqn:2})-(\ref{eqn:5}) holds, it
then follows that  $$m<v+\delta-1~\text{and}~u\leq 2(r-v).$$ In
that case, if condition (\ref{eqn:6}) does not hold, we have
$w<r+\delta-1-m~\text{or}~r-v<u$; and if condition (\ref{eqn:7})
does not hold, we have $w+1<2(r+\delta-1-m)$, i.e.,
$w<2(r+\delta-1-m)-1$. Hence, if, neither condition (\ref{eqn:6})
nor condition (\ref{eqn:7}) holds (in addition to
(\ref{eqn:2})-(\ref{eqn:5})), then one of the following two
conditions must be satisfied:
\begin{align}
w<r+\delta-1-m, \label{eqn:8}
\end{align} or
\begin{align}
r+\delta-1-m\leq w<2(r+\delta-1-m)-1\ \text{and}\ r-v<u.\label{eqn:9}
\end{align}
In other words, if none of the conditions
(\ref{eqn:2})-(\ref{eqn:7}) holds, then either (\ref{eqn:8}) or
(\ref{eqn:9}) will hold. From our existence proof and/or
constructive results, the existence of optimal $(r,\delta)_a$
linear code is not known only for a limited scope with parameters
described by (\ref{eqn:8}) and (\ref{eqn:9}).

The remainder of the paper is organized as follows. In Section
\uppercase\expandafter{\romannumeral2}, we present the notions
used in the paper as well as some preliminary results about
$(r,\delta)_a$ linear codes.
% of symbol locality for linear code introduced by Prakash \emph{et
%al.} \cite{Prakash12}, and give some simple properties of
%$(r,\delta)_a$ linear codes.
 In Section
\uppercase\expandafter{\romannumeral3}, we investigate the
structure of optimal $(r,\delta)_a$ linear codes when $r|k$
(should they exist). In Section
\uppercase\expandafter{\romannumeral4}, we consider the
non-existence conditions for optimal $(r,\delta)_a$ linear codes
under conditions (\ref{eqn:2}) and (\ref{eqn:3}). A construction
of optimal $(r,\delta)_a$ linear codes for conditions
(\ref{eqn:4}) and (\ref{eqn:5}) is presented in Section
\uppercase\expandafter{\romannumeral5}, and a construction of
optimal $(r,\delta)_a$ linear codes for conditions (\ref{eqn:6})
and (\ref{eqn:7}) is presented in Section
\uppercase\expandafter{\romannumeral6}. Finally, we conclude the
paper in Section \uppercase\expandafter{\romannumeral7}.

\section{Locality of Linear Codes}

For two positive integers $t_1$ and $t_2 ~(t_1\leq t_2)$, we
denote $[t_1,t_2]=\{t_1,t_1+1,\cdots,t_2\}$ and
$[t_2]=\{1,2,\cdots,t_2\}$. For any set $S$, the size
$($cardinality$)$ of $S$ is denoted by $|S|$. If $I$ is a subset
of $S$ and $|I|=r$, then we say that $I$ is an $r$-subset of $S$.
Let $\mathbb F_q^k$ be the $k$-dimensional vector space over the
$q$-ary field $\mathbb F_q$. For any subset $X\subseteq\mathbb
F_q^k$, we use $\langle X\rangle$ to denote the subspace of
$\mathbb F_q^k$ spanned by $X$.

In the sequel, whenever we speak of an $(r,\delta)_a$ or
$(r,\delta)_i$ code, we will by default assume it is an $[n,k,d]$
linear code (i.e., its length, dimension and minimum distance are
$n,k$ and $d$ respectively).

Suppose $\mathcal C$ is an $[n,k,d]$ linear code over $\mathbb
F_q$, and $G=(G_1,\cdots,G_n)$ is a generating matrix of $\mathcal
C$, where $G_i, i\in[n],$ is the $i$th column of $G$. We denote by
$\mathcal G=\{G_1,\cdots,G_n\}$ the collection of columns of $G$.
It is well known that the distance property is captured by the
following condition (e.g. \cite{Tsfasman}).

\vskip 10pt
\begin{lem}\label{fact}
An $[n,k]$ code $\mathcal C$ has a minimum distance $d$, if and
only if $|S|\leq n-d$ for every $S\subseteq\mathcal G$ having
$\text{Rank}(S)\leq k-1$. Equivalently, $\text{Rank}(T)=k$ for
every $T\subseteq\mathcal G$ of size $n-d+1$.
\end{lem}
\vskip 10pt

For any subset $S\subseteq[n]$, let $\mathcal C|_{S}$ denote the
punctured code of $\mathcal C$ associated with the coordinate set
$S$. That is, $\mathcal C|_{S}$ is obtained from $\mathcal C$ by
deleting all symbols $c_i, i\in[n]\backslash S$, in each codeword
$(c_1,\cdots,c_n)\in\mathcal C$.

\vskip 10pt
\begin{defn}[\cite{Prakash12}]\label{def-locality}
Suppose $1\leq r\leq k$ and $\delta\geq 2$. The $i$th code symbol
$c_i, 1\leq i\leq n$, in an $[n,k,d]$ linear code $\mathcal C$ is
said to have locality $(r,\delta)$ if there exists a subset
$S_i\subseteq[n]$ such that
\begin{itemize}
    \item [(1)] $|S_i|\leq r+\delta-1$;
    \item [(2)] The minimum distance of the punctured code
    $\mathcal C|_{S_i}$ is at least $\delta$.
\end{itemize}
\end{defn}
\vskip 10pt

\begin{rem}\label{rem-loty}
Let $G=(G_1,\cdots,G_n)$ be a generating matrix of $\mathcal C$.
By Lemma \ref{fact}, it is easy to see that the second condition in
Definition \ref{def-locality} is equivalent to the following
condition
\begin{itemize}
    \item [(2$'$)] $\text{Rank}(\{G_\ell; \ell\in I\})=
    \text{Rank}(\mathcal G_i)$ for any subset $I\subseteq S_i$ of size
    $|I|=|S_i|-\delta+1$, where $\mathcal G_i=\{G_\ell; \ell\in S_i\}$;
\end{itemize}
\end{rem}
\vskip 10pt

Moreover, by conditions (1) and (2$'$), we have
$$\text{Rank}(\mathcal G_i)=\text{Rank}(\{G_\ell; \ell\in
S_i\})\leq|S_i|-\delta+1\leq r.$$
That is, $\forall i'\in S_i$ and
$\forall I\subseteq S_i\backslash\{i'\}$ of size
$|I|=|S_i|-\delta+1$, $G_{i'}$ is an $\mathbb F_q$-linear
combination of $\{G_\ell;\ell\in I\}$. This means that the symbol
$c_{i'}$ can be reconstructed by the $|S_i|-\delta+1$ symbols in
$\{c_\ell;\ell\in I\}$.

%A linear code $\mathcal C$ is called an $(r,\delta)_a$ code if all
%code symbols of $\mathcal C$ have locality $(r,\delta)$. If
%$\mathcal C$ is a systematic code and all information symbols have
%locality $(r,\delta)$, then $\mathcal C$ is called an
%$(r,\delta)_i$ code.
 An $(r,\delta)_a$ code $\mathcal C$ is said to be
\emph{optimal} if the minimum distance $d$ of $\mathcal C$
achieves the bound in (\ref{eqn:1}).

The following remark follows naturally from Definition
\ref{def-locality} and Remark \ref{rem-loty}.
%, the following remark is obvious.

\vskip 10pt
\begin{rem}\label{rem-locality}
If $\mathcal C$ is an $(r,\delta)_a$ code and $G=(G_1,\cdots,G_n)$
is a generating matrix of $\mathcal C$, then we can always find a
collection $\mathcal S=\{S_1,\cdots, S_t\}$, where
$S_i\subseteq[n], i=1,\cdots,t$, such that
\begin{itemize}
  \item [(1)] $|S_i|\leq r+\delta-1, i=1,\cdots,t$;
  \item [(2)] $\text{Rank}(\{G_\ell; \ell\in I\})=
  \text{Rank}(\mathcal G_i)\leq r,
  \forall i\in[t]$ and $I\subseteq S_i$ of size $|I|=|S_i|-\delta+1$,
  where $\mathcal G_i=\{G_\ell; \ell\in S_i\}$;
  \item [(3)] $\cup_{i\in[t]}S_i=[n]$ and $\cup_{i\in[t]\backslash\{j\}}
  S_i\neq[n],\forall j\in[t]$.
\end{itemize}
We call the set $\mathcal S=\{S_1,\cdots,S_t\}$ an
$(r,\delta)$-\emph{cover set} of $\mathcal C$.
\end{rem}
\vskip 10pt

The following lemma presents a simple property of $(r,\delta)_a$ codes.

\vskip 10pt
\begin{lem}\label{low-bound}
An $(r,\delta)_a$ code $\mathcal C$  satisfies
\begin{itemize}
  \item [1)] The minimum distance $d\geq\delta$.
  \item [2)] If $\mathcal C$ is an optimal $(r,\delta)_a$ code, then
  $\frac{n}{r+\delta-1}\geq\frac{k}{r}$.
\end{itemize}
\end{lem}
\begin{proof}
1) Let $\mathcal S=\{S_1,\cdots,S_t\}$ be an $(r,\delta)$-cover
set of $\mathcal C$. For any $0\neq(c_1,\cdots,c_n)\in\mathcal C$,
since $\cup_{i\in[t]}S_i=[n]$, there is an $i\in[t]$ such that the
punctured codeword $(c_j)_{j\in S_i}$ is nonzero in $\mathcal
C|_{S_i}$. By the second condition of Definition \ref{def-locality}, the Hamming
weight of $(c_j)_{j\in S_i}$ is at least $\delta$. Thus, the Hamming
weight of $(c_1,\cdots,c_n)$ is at least $\delta$. Since
$0\neq(c_1,\cdots,c_n)\in\mathcal C$ is arbitrary, the minimum
distance $d\geq\delta$.

2) Since $\mathcal C$ is an optimal $(r,\delta)_a$ code, from the minimum distance bound in (\ref{eqn:1}),
$$n=d+k-1+\left(\left\lceil\frac{k}{r}\right\rceil-1\right)(\delta-1).$$
From claim 1),
$d\geq\delta$; which leads to
$$n\geq\delta+k-1+\left(\left\lceil\frac{k}{r}\right\rceil-1\right)(\delta-1).$$ Hence,
\vspace{-0.05in}\begin{eqnarray*} nr&\geq&
r(\delta+k-1)+r(\lceil\frac{k}{r}\rceil-1)(\delta-1)\\
&\geq&r(\delta+k-1)+r(\frac{k}{r}-1)(\delta-1)\\
&=&k(r+\delta-1)
\end{eqnarray*}
which implies that $\frac{n}{r+\delta-1}\geq\frac{k}{r}$.
\end{proof}

\section{Structure of Optimal $(r,\delta)_a$ Code when $r|k$}
In this section, we prove a structure theorem for optimal
$(r,\delta)_a$ codes under the condition of $r|k$.
% Note that $n,k$ and $d$ are the default value of the length, dimension and minimum distance of the linear code respectively.

Throughout this section, we assume that $\mathcal C$ is an
$(r,\delta)_a$ code over the field $\mathbb F_q$,  $\mathcal
S=\{S_1,\cdots,S_t\}$ is an $(r,\delta)$-cover set of $\mathcal
C$, where $S_i\subseteq[n]$, $i=1,\cdots t$, and
$G=(G_1,\cdots,G_n)$ is a generating matrix of $\mathcal C$. We
denote $\mathcal G=\{G_1,\cdots,G_n\}$ and $\mathcal G_i=\{G_\ell;
\ell\in S_i\}$\footnote{When $G_i$ and $G_j$ are viewed as vectors
of $\mathbb F_q^k$, it is possible for $G_i=G_j$ where $i\neq j$.
However, when treating them as two different columns of $G$, we
shall view $G_i$ and $G_j$ as two separate elements in $\mathcal
G$ (even though they may be identical).}. Then for any
$I\subseteq[t]$, we have
\begin{align}
|\cup_{i\in I}\mathcal G_i|=|\{G_i; i\in\cup_{\ell\in I}S_\ell\}|
=|\cup_{i\in I}S_i| \label{eqn:10}
\end{align} and by
Remark \ref{rem-locality}, we get
\begin{align}
\cup_{i\in[t]}\mathcal G_i=\mathcal G ~\text{and} ~
\cup_{i\in[t]\backslash\{j\}}\mathcal G_i\neq\mathcal G, \forall
j\in[t].\label{eqn:11}
\end{align}

We first give some lemmas to help prove our main results.

\vskip 10pt
\begin{lem}\label{rank-sum}
Consider three sets $A,B,X\subseteq\mathbb F_q^k$. If $C$ is a
subset of $X$ satisfies: $\text{Rank}(B\cup C)=\text{Rank}(A\cup
B\cup C),$ then
$$\text{Rank}(X\cup A\cup B)-|B|\leq\text{Rank}(X).$$
\end{lem}
\begin{proof}
Since $C\subseteq X$ and $\text{Rank}(B\cup C)=\text{Rank}(A\cup
B\cup C)$, we have \vspace{-0.05in}\begin{eqnarray*}
\text{Rank}(X\cup A\cup B)&=&\text{Rank}(X\cup C\cup A\cup B)\\
&=&\text{Rank}(X\cup B\cup C)\\&=&\text{Rank}(X\cup B)\\&\leq&
\text{Rank}(X)+\text{Rank}(B)\\&\leq&\text{Rank}(X)+|B|.
\end{eqnarray*}
Therefore, $\text{Rank}(X\cup A\cup B)-|B|\leq\text{Rank}(X)$.
\end{proof}
\vskip 10pt

\begin{lem}\label{cup-rank}
Suppose $\{i_1,\cdots,i_\ell\}\subseteq[t]$ such that $\mathcal
G_{i_j}\nsubseteq\langle\cup_{\lambda=1}^{j-1}\mathcal
G_{i_\lambda}\rangle$, $j=2,\cdots,\ell$. Then
$$|\cup_{j=1}^{\ell}S_{i_j}|\geq\text{Rank}(\cup_{j=1}^{\ell}\mathcal
G_{i_j})+\ell(\delta-1).$$
\end{lem}
\begin{proof} We prove this lemma by induction.

From Remark \ref{rem-loty}, $|S_{i_1}|\geq\text{Rank}(\mathcal
G_{i_1})+(\delta-1)$. Hence the claim holds for $\ell=1$.

Now consider $\ell\ge 2$. We assume  that the claim holds for $\ell-1$, i.e.,
\begin{align}
|\cup_{j=1}^{\ell-1}S_{i_j}|\geq\text{Rank}(\cup_{j=1}^{\ell-1}\mathcal
G_{i_j})+(\ell-1)(\delta-1).\label{eqn:12}
\end{align}
We shall prove that the claim is true for $\ell$.

First, we point out that $|\mathcal
G_{i_\ell}\backslash(\cup_{j=1}^{\ell-1}\mathcal
G_{i_j})|>\delta-1.$ In fact, if $|\mathcal
G_{i_\ell}\backslash(\cup_{j=1}^{\ell-1}\mathcal
G_{i_j})|\leq\delta-1$, then $|\mathcal
G_{i_\ell}\cap(\cup_{j=1}^{\ell-1}\mathcal G_{i_j}|\geq|\mathcal
G_{i_\ell}|-(\delta-1)$. From condition (2) of Remark \ref{rem-locality},
$\mathcal G_{i_\ell}\subseteq\langle\mathcal
G_{i_\ell}\cap(\cup_{j=1}^{\ell-1}\mathcal
G_{i_j})\rangle\subseteq \langle\cup_{j=1}^{\ell-1}\mathcal
G_{i_j}\rangle$, which presents a contradiction to the assumption that $\mathcal
G_{i_\ell}\nsubseteq\langle\cup_{j=1}^{\ell-1}\mathcal
G_{i_j}\rangle$. Thus,
$$|\mathcal G_{i_\ell}\backslash(\cup_{j=1}^{\ell-1}\mathcal G_{i_j})|>\delta-1.$$

Let $X=\cup_{j=1}^{\ell-1}\mathcal G_{i_j}$ and $C=\mathcal
G_{i_\ell}\cap(\cup_{j=1}^{\ell-1}\mathcal G_{i_j})=\mathcal
G_{i_\ell}\cap X$. Let $A$ be a fixed $(\delta-1)$-subset of
$\mathcal G_{i_\ell}\backslash(\cup_{j=1}^{\ell-1}\mathcal
G_{i_j})$ and $B=(\mathcal
G_{i_\ell}\backslash\cup_{j=1}^{\ell-1}\mathcal G_{i_j})\backslash
A$.

From condition (2) of Remark \ref{rem-locality}, $\text{Rank}(B\cup
C)=\text{Rank}(A\cup B\cup C).$ Then, from Lemma \ref{rank-sum}, we
get
$$\text{Rank}(X\cup A\cup B)-|B|\leq\text{Rank}(X)$$ i.e.,
\begin{align}
\text{Rank}(\cup_{j=1}^{\ell}\mathcal G_{i_j})-|B|\leq
\text{Rank}(\cup_{j=1}^{\ell-1}\mathcal G_{i_j}).\label{eqn:13}
\end{align}
Clearly, $\cup_{j=1}^{\ell}\mathcal G_{i_j}$ is a disjoint union
of $A, B$ and $\cup_{j=1}^{\ell-1}\mathcal G_{i_j}$. Hence,
\begin{eqnarray*} |\cup_{j=1}^{\ell}\mathcal
G_{i_j}|&=&|\cup_{j=1}^{\ell-1}\mathcal
G_{i_j}|+|A|+|B|\nonumber\\&=&|\cup_{j=1}^{\ell-1}\mathcal
G_{i_j}|+(\delta-1)+|B|\nonumber
\end{eqnarray*}
and from (\ref{eqn:10}), we get
\begin{align}
|\cup_{j=1}^{\ell}S_{i_j}|=|\cup_{j=1}^{\ell}\mathcal G_{i_j}|
=|\cup_{j=1}^{\ell-1}S_{i_j}|+(\delta-1)+|B|.\label{eqn:14}
\end{align}
Combining (\ref{eqn:12})-(\ref{eqn:14}), we have
\begin{eqnarray*} |\cup_{j=1}^{\ell}S_{i_j}|
&=&|\cup_{j=1}^{\ell-1}S_{i_j}|+(\delta-1)+|B|
\\&\geq&\text{Rank}(\cup_{j=1}^{\ell-1}\mathcal
G_{i_j})+\ell(\delta-1)+|B|\\
&\geq&\text{Rank}(\cup_{j=1}^{\ell}\mathcal
G_{i_j})-|B|+\ell(\delta-1)+|B|\\&=&\text{Rank}(\cup_{j=1}^{\ell}\mathcal
G_{i_j})+\ell(\delta-1)
\end{eqnarray*}
which completes the proof.
\end{proof}
\vskip 10pt

\begin{lem}\label{not-ctn}
Suppose $\mathcal C$ is an optimal $(r,\delta)_a$ code. Then
\begin{itemize}
  \item [1)] $t\geq\lceil\frac{n}{r+\delta-1}\rceil\geq\lceil\frac{k}{r}\rceil$.
  \item [2)] If $J\subseteq[t]$ and $|J|\leq\lceil\frac{k}{r}\rceil-1$,
  then $\text{Rank}(\cup_{i\in J}\mathcal G_i)\leq k-1$ and
  $\mathcal G_h\nsubseteq\langle\cup_{i\in
  J}\mathcal G_i\rangle, \forall h\in[t]\backslash J$.
  \item [3)] If $J\subseteq[t]$ and $|J|=\lceil\frac{k}{r}\rceil$,
  then $\text{Rank}(\cup_{i\in J}\mathcal G_i)=k$ and
  $|\cup_{i\in J}S_i|\geq k+\lceil\frac{k}{r}\rceil(\delta-1)$.
\end{itemize}
\end{lem}
\begin{proof}
1) (Proof by contradiction) Suppose $t\leq\lceil\frac{n}{r+\delta-1}\rceil-1$. Then from
Remark \ref{rem-locality}, $$|S_i|\leq r+\delta-1.$$ Hence,
\begin{eqnarray*} n&=&|\cup_{i\in[t]}S_i|\\
&\leq&t(r+\delta-1)\\
&\leq&(\lceil\frac{n}{r+\delta-1}\rceil-1)(r+\delta-1)\\&<&n
\end{eqnarray*} which presents a contradiction. Hence, it must hold
 that $t\geq\lceil\frac{n}{r+\delta-1}\rceil$.

Moreover, from Claim 2) of Lemma \ref{low-bound},
$\frac{n}{r+\delta-1}\geq\frac{k}{r}$. Thus,
$$t\geq\lceil\frac{n}{r+\delta-1}\rceil\geq\lceil\frac{k}{r}\rceil.$$

2) From Remark \ref{rem-loty}, $\text{Rank}(\mathcal G_i)\leq r,
\forall i\in[t]$. Hence, if $|J|\leq\lceil\frac{k}{r}\rceil-1$,
then
$$\text{Rank}(\cup_{i\in J}\mathcal G_i)\leq r|J|\leq
r(\lceil\frac{k}{r}\rceil-1)<r\frac{k}{r}=k.$$ i.e.,
$\text{Rank}(\cup_{i\in J}\mathcal G_i)\leq k-1.$

Now, suppose $\mathcal G_h\subseteq\langle\cup_{i\in J}\mathcal
G_i\rangle$, and we will see a contradiction results. First, we can find a subset
$J_0=\{i_1,\cdots,i_s\}\subseteq J$ such that $\mathcal
G_h\subseteq\langle\cup_{\lambda=1}^{s}\mathcal G_{i_s}\rangle$
and $\mathcal G_h\nsubseteq\langle\cup_{i\in J'}\mathcal
G_i\rangle$ for any proper subset $J'$ of $J_0$. In particular, we
have
$$\mathcal G_{i_j}\nsubseteq\langle\cup_{\lambda=1}^{j-1}\mathcal
G_{i_\lambda}\rangle, j=2,\cdots,s.$$ Note that $|J_0|\leq
|J|\leq\lceil\frac{k}{r}\rceil-1$. By the proved result, we have
$$\text{Rank}(\cup_{i\in J_0}\mathcal G_i)\leq k-1.$$
Next, we can find a sequence $\mathcal G_{i_1},\cdots,\mathcal
G_{i_s}, \mathcal G_{i_{s+1}},\cdots, \mathcal G_{i_\ell}$ such
that $\ell\geq\lceil\frac{k}{r}\rceil,
\text{Rank}(\cup_{j=1}^\ell\mathcal G_{i_j})=k$ and $\mathcal
G_{i_j}\nsubseteq\langle\cup_{\lambda=1}^{j-1}\mathcal
G_{i_\lambda}\rangle, j=2,\cdots,\ell$. In particular,
$\text{Rank}(\cup_{j=1}^{\ell-1}\mathcal G_{i_j})\leq k-1$.
Therefore, there exists a $\mathcal G'_{i_\ell}\subseteq\mathcal
G_{i_\ell}$ such that $\text{Rank}((\cup_{j=1}^{\ell-1}\mathcal
G_{i_j})\cup\mathcal G'_{i_\ell})=k-1$. Denote
$(\cup_{j=1}^{\ell-1}\mathcal G_{i_j})\cup\mathcal G'_{i_\ell}=S$.
Then $\text{Rank}(S)=k-1$ and
\begin{eqnarray}
\nonumber |\mathcal
G'_{i_\ell}\backslash\cup_{j=1}^{\ell-1}\mathcal
G_{i_j}|&\geq&\text{Rank}(S)-\text{Rank}(\cup_{j=1}^{\ell-1}\mathcal
G_{i_j})\\
&=&(k-1)-\text{Rank}(\cup_{j=1}^{\ell-1}\mathcal
G_{i_j}). \label{eqn:15}
\end{eqnarray}
From Lemma \ref{cup-rank},
\begin{align}
|\cup_{j=1}^{\ell-1}\mathcal
G_{i_j}|\geq\text{Rank}(\cup_{j=1}^{\ell-1}\mathcal
G_{i_j})+(\ell-1)(\delta-1).
\label{eqn:16}
\end{align}
Then by equations (\ref{eqn:15}) and (\ref{eqn:16}),
\begin{eqnarray} \nonumber|S|&=&|\mathcal
G'_{i_\ell}\backslash\cup_{j=1}^{\ell-1}\mathcal
G_{i_j}|+|\cup_{j=1}^{\ell-1}\mathcal G_{i_j}|\\
\nonumber&\geq&(k-1)+(\ell-1)(\delta-1)\\
&\geq&k-1+(\lceil\frac{k}{r}\rceil-1)(\delta-1). \label{eqn:17}
\end{eqnarray}
Since $h\in[t]\backslash J$, $\mathcal G_h\neq \mathcal G_{i_j},
j=1,\cdots,s$. Moreover, since $\mathcal
G_h\subseteq\langle\cup_{\lambda=1}^{s}\mathcal G_{i_s}\rangle$
and $\mathcal
G_{i_j}\nsubseteq\langle\cup_{\lambda=1}^{j-1}\mathcal
G_{i_\lambda}\rangle, j=2,\cdots,\ell$, so $\mathcal G_h\neq
\mathcal G_{i_j}, j=s+1,\cdots,\ell$. From equation (\ref{eqn:11}), we have
$\mathcal G_h\nsubseteq\cup_{j=1}^\ell \mathcal G_{i_j}$. Then, from
equation (\ref{eqn:17}), we get
$$|\mathcal G_h\cup S|>|S|\geq
k-1+(\lceil\frac{k}{r}\rceil-1)(\delta-1).$$ Since we assumed
$\mathcal G_h\subseteq\langle\cup_{\lambda=1}^{s}\mathcal
G_{i_s}\rangle\subseteq\langle S\rangle$, then
$\text{Rank}(\mathcal G_h\cup S)=\text{Rank}(S)=k-1$. By Lemma
\ref{fact}, we have
$$d\leq n-|\mathcal G_h\cup
S|<n-k+1-(\lceil\frac{k}{r}\rceil-1)(\delta-1),$$ which
contradicts the assumption that $\mathcal C$ is an optimal
$(r,\delta)_a$ code. Hence, it must be that $\mathcal
G_h\nsubseteq\langle\cup_{i\in J}\mathcal G_i\rangle$.\footnote{In
this proof, for any $(r,\delta)_a$ code $\mathcal C$, we obtain a
subset $S\subseteq\mathcal G$ such that $|S|\geq
k-1+(\lceil\frac{k}{r}\rceil-1)(\delta-1)$ and Rank$(S)=k-1$. Then
by Lemma \ref{fact}, the minimum distance of $\mathcal C$ is
$d\leq n-k+1-(\lceil\frac{k}{r}\rceil-1)(\delta-1)$, which also
provides a proof of the minimum distance bound in (\ref{eqn:1}).}

3) Suppose $J=\{i_1,\cdots,i_s\}$, where
$s=\lceil\frac{k}{r}\rceil$. By claim 2), $$\mathcal
G_{i_j}\nsubseteq\langle\cup_{\lambda=1}^{j-1}\mathcal
G_{i_\lambda}\rangle, j=2,\cdots,s.$$

First, we have $\text{Rank}(\cup_{i\in J}\mathcal G_i)=k$.
Otherwise, as in the proof of claim 2), we can find a sequence
$\mathcal G_{i_1},\cdots,\mathcal G_{i_s}$, $\mathcal
G_{i_{s+1}},\cdots, \mathcal
G_{i_\ell}~(\ell>s=\lceil\frac{k}{r}\rceil)$ and a set
$S=(\cup_{j=1}^{\ell-1}\mathcal G_{i_j})\cup\mathcal
G'_{i_\ell}~(\mathcal G'_{i_\ell}\subseteq\mathcal G_{i_\ell})$
such that $$|S|\geq k-1+(\ell-1)(\delta-1)>
k-1+(\lceil\frac{k}{r}\rceil-1)(\delta-1).$$ By Lemma \ref{fact},
$$d\leq n-|S|<n-k+1-(\lceil\frac{k}{r}\rceil-1)(\delta-1)$$ which
contradicts the assumption that $\mathcal C$ is an optimal
$(r,\delta)_a$ code. Therefore,  we have $\text{Rank}(\cup_{i\in
J}\mathcal G_i)=k$.

Now, by Lemma \ref{cup-rank}, \begin{eqnarray*} |\cup_{i\in
J}S_i|&\geq&\text{Rank}(\cup_{i\in J}\mathcal
G_i)+\lceil\frac{k}{r}\rceil(\delta-1)
\\&=&k+\lceil\frac{k}{r}\rceil(\delta-1).
\end{eqnarray*}
This completes the proof.
\end{proof}
\vskip 10pt

We now  present our main theorem of this section.

\vskip 10pt
\begin{thm}\label{stru-opt}
Suppose $\mathcal C$ is an optimal $(r,\delta)_a$ linear code. If
$r|k$ and $r<k$, then the following conditions hold:
\begin{itemize}
\item [1)] $S_1,\cdots,S_t$ are mutually disjoint; \item [2)]
$|S_i|=r+\delta-1, \forall i\in[t]$, and the punctured code
$\mathcal C|_{S_i}$ is an $[r+\delta-1,r,\delta]$ MDS code.
\end{itemize}
In particular, we have $(r+\delta-1)\mid n$.
\end{thm}
\begin{proof}
Since $r|k$ and $r<k$, then $k=\ell r$ for some $\ell\geq 2$. By
1) of Lemma \ref{not-ctn}, $t\geq\lceil\frac{k}{r}\rceil=\ell$.
Let $\{i_1,i_2\}\subseteq[t]$ be arbitrarily chosen. Let $J$ be an
$\ell$-subset of $[t]$ such that $\{i_1,i_2\}\subseteq J$. Then by
3) of Lemma \ref{not-ctn},
\begin{align}
\text{Rank}(\cup_{i\in J}\mathcal G_i)=k=\ell r, \label{eqn:18}
\end{align} and
\begin{align}
\left|\cup_{i\in J}\mathcal S_{i}\right|\geq k+\ell(\delta-1)=
\ell(r+\delta-1).\label{eqn:19}
\end{align}
Since $|S_{i}|\leq r+\delta-1$ and by Remark \ref{rem-locality},
$\text{Rank}(\mathcal G_i)\leq r$, then equations (\ref{eqn:18})
and (\ref{eqn:19}) imply that $\text{Rank}(\mathcal G_i)=r$,
$|S_{i}|=r+\delta-1$, and $\{S_{i}\}_{i\in J}$ are mutually
disjoint.

In particular, $\text{Rank}(\mathcal G_{i_1})=\text{Rank}(\mathcal
G_{i_2})=r$, $\mathcal G_{i_1}\cap\mathcal G_{i_2}=\emptyset$ and
$|S_{i_1}|=|S_{i_2}|=r+\delta-1$. Since $i_1$ and $i_2$ are
arbitrarily chosen, we have proved that $\text{Rank}(\mathcal
G_i)=r$, $|S_{i}|=r+\delta-1$, and $\{S_{i}\}_{i\in J}$ are
mutually disjoint. Hence, $(r+\delta-1)\mid n$. Moreover, by Lemma
\ref{fact} and Remark \ref{rem-loty}, $\mathcal C|_{S_i}$ is an
$[r+\delta-1,r,\delta]$ MDS code.
\end{proof}
\vskip 10pt

In \cite{Prakash12}, it was proved that if $\mathcal C$ is an
optimal $(r,\delta)_i$ code, then there exists a collection
$\{S_1,\cdots,S_a\}\subseteq\{S_1,\cdots,S_t\}$ which has the same
properties in Theorem \ref{stru-opt}, where $a$ is a
properly-defined value. Thus, Theorem \ref{stru-opt} shows that as
a sub-class of optimal $(r,\delta)_i$ codes, optimal
$(r,\delta)_a$ codes tend to have a simpler structure than
otherwise.

\section{Non-existence Conditions of Optimal $(r,\delta)_a$ Linear Codes}
In this section, we derive two sets of conditions under which
there exists no optimal $(r,\delta)_a$ linear codes. From the
minimum distance bound in (\ref{eqn:1}), we know that when $r=k$,
optimal $(r,\delta)_a$ linear codes are exactly MDS codes. Hence,
in this section, we focus on the case of $r<k$.

The first result is obtained directly from Theorem \ref{stru-opt}.

\vskip 10pt
\begin{thm}\label{non-exst}
If $(r+\delta-1)\nmid n$ and $r|k~$, then there exist no
optimal $(r,\delta)_a$ linear codes.
\end{thm}
\begin{proof}
If $\mathcal C$ is an optimal $(r,\delta)_a$ linear code and
$r|k$, then by Theorem \ref{stru-opt}, $(r+\delta-1)|n$, which
contradicts the condition that $(r+\delta-1)\nmid n$. Hence, there
exist no optimal $(r,\delta)_a$ linear codes when
$(r+\delta-1)\nmid n$ and $r|k$.
\end{proof}
\vskip 10pt

When $(r+\delta-1)\nmid n$ and $r\nmid k$, we provide in the below
a set of conditions under which no optimal $(r,\delta)_a$ code
exists.

\vskip 10pt
\begin{thm}\label{non-exst-1}
Suppose $n=w(r+\delta-1)+m$ and $k=ur+v$, where $0<m<r+\delta-1$
and $0<v<r$. If $m<v+\delta-1$ and $u\geq 2(r-v)+1$, then there
exist no optimal $(r,\delta)_a$ codes.
\end{thm}
\begin{proof}
We prove this theorem by contradiction.

Suppose $\mathcal C$ is an optimal $(r,\delta)_a$ code over the
field $\mathbb F_q$ and $\mathcal S=\{S_1,\cdots,S_t\}$ is an
$(r,\delta)$-cover set of $\mathcal C$. Then by claim 1) of Lemma
\ref{not-ctn}, we have
\begin{align}
t\geq\left\lceil\frac{n}{r+\delta-1}\right\rceil=w+1.\label{eqn:20}
\end{align}
Moreover, by 3) of Lemma \ref{not-ctn}, for any
$\lceil\frac{k}{r}\rceil$-subset $J$ of $[t]$, $$|\cup_{i\in
J}S_i|\geq k+\left\lceil\frac{k}{r}\right\rceil(\delta-1).$$

For each $i\in[t]$, if $|S_i|<r+\delta-1$, let $T_i\subseteq[n]$
be such that $S_i\subseteq T_i$ and $|T_i|=r+\delta-1$; If
$|S_i|=r+\delta-1$, let $T_i=S_i$. Then clearly,
$$\cup_{i\in[t]}T_i=\cup_{i\in[t]}S_i=[n]$$ and for any
$\lceil\frac{k}{r}\rceil$-subset $J$ of $[t]$,
\begin{align} |\cup_{i\in
J}T_i|\geq k+\left\lceil\frac{k}{r}\right\rceil(\delta-1).\label{eqn:21}
\end{align}

%We now prove that if $m<v+\delta-1$ and $u\geq 2(r-v)+1$, then we
%can find a $J_0\subseteq[t]$ such that $$|\cup_{i\in J_0}T_i|\leq
%k-1+\lceil\frac{k}{r}\rceil(\delta-1)$$ which contradicts to (19).
%Thus,

Let $M=(m_{i,j})_{t\times n}$ be a $t\times n$ matrix such that
$m_{i,j}=1$ if $j\in T_i$, and $m_{i,j}=0$ otherwise. For each
$j\in[n]$, let
$$A_j=\{i\in[t];m_{i,j}=1\}.$$ Then $|A_j|$ is the number of $T_i
~(i\in[t])$ satisfying $j\in T_i$, and this number equals  the
number of $1$s in the $j$th column of $M$. Since
$\cup_{i\in[t]}T_i=[n]$, then $|A_j|>0, \forall j\in[n]$. On the
other hand, by the construction of $M$, for each $i\in[t]$,
$T_i=\{j\in[n];m_{i,j}=1\}.$ Thus, the number of the $1$s in each
row of $M$ is $r+\delta-1$. It then follows that the total number
of the $1$s in $M$ is
\begin{align}
\sum_{j=1}^n|A_j|=\sum_{i=1}^t|T_i|=t(r+\delta-1).\label{eqn:22}
\end{align}
Combining (\ref{eqn:20}) and (\ref{eqn:22}), we have
\begin{align} \nonumber
\sum_{j=1}^n|A_j|\geq&(w+1)(r+\delta-1)\\
=&n+(r+\delta-1-m).\label{eqn:23}
\end{align}
Since $m<v+\delta-1$, then $$r+\delta-1-m>r-v.$$ Hence from (\ref{eqn:23}), we
have
\begin{align}
\sum_{j=1}^n|A_j|\geq n+(r-v+1).\label{eqn:24}
\end{align}

Let $P=\{j\in[n];|A_j|>1\}$. From (\ref{eqn:24}), $P\neq\emptyset$ and
$$\sum_{j\in P}|A_j|\geq|P|+(r-v+1).$$
Without loss of generality, assume $P=\{1,\cdots,\ell\}$. Since
$|A_j|>1, \forall j\in P$, we can find a number
$\lambda\in\{1,\cdots,\ell\}$ such that
$\sum_{j=1}^{\lambda-1}|A_{j}|<\lambda+(r-v)$ and
$\sum_{j=1}^{\lambda}|A_{j}|\geq\lambda+(r-v+1)$. This means that
we can find a subset $B_{\lambda}\subseteq A_{\lambda}$ such that
$|B_{\lambda}|>1$ and
\begin{align}
\sum_{j=1}^{\lambda-1}|A_{j}|+|B_{\lambda}|=\lambda+r-v+1.\label{eqn:25}
\end{align}
Also note that
\begin{align}
\lambda\leq r-v+1,\label{eqn:26}
\end{align} because otherwise,
$\sum_{j=1}^{\lambda-1}|A_{j}|+|B_{\lambda}|\geq
2\lambda>\lambda+r-v+1$, which contradicts (\ref{eqn:25}).

Let $B=(\cup_{j=1}^{\lambda-1}A_{j})\cup B_\lambda$. Then from
(\ref{eqn:25}),
$$|B|=|(\cup_{j=1}^{\lambda-1}A_{j})\cup B_\lambda|
\leq\sum_{j=1}^{\lambda-1}|A_{i}|+|B_{\lambda}|\leq 2(r-v+1).$$
Since $u\geq 2(r-v)+1$, then $2(r-v+1)\leq u+1$, we get
$$|B|\leq u+1=\left\lceil\frac{k}{r}\right\rceil.$$
Let $J$ be a $\lceil\frac{k}{r}\rceil$-subset of $[t]$ such that
$B\subseteq J$. By the construction of $M$ and $B$, for each
$j\in\{1,\cdots,\lambda-1\}$, there are at least $|A_j|$ subsets
in $\{T_i;i\in B\}$ containing $j$, and there are at least
$|B_\lambda|$ subsets in $\{T_i;i\in B\}$ containing $\lambda$. Hence,
\begin{align}
|\cup_{i\in
J}T_i|\leq|J|(r+\delta-1)-(\sum_{j=1}^{\lambda-1}|A_{j}|+
|B_{\lambda}|-\lambda).\label{eqn:27}
\end{align}
Combining (\ref{eqn:25}) and (\ref{eqn:27}), we
have
\begin{eqnarray*} |\cup_{i\in
J}T_i|&\leq&\lceil\frac{k}{r}\rceil(r+\delta-1)-(r-v+1)\\
&=&ur+v-1+\lceil\frac{k}{r}\rceil(\delta-1)\\
&=&k-1+\lceil\frac{k}{r}\rceil(\delta-1).
\end{eqnarray*}
which contradicts (\ref{eqn:21}).

Thus, we can conclude that there exist no optimal $(r,\delta)_a$
linear codes when $m<v+\delta-1$ and $u\geq 2(r-v)+1$.
\end{proof}
\vskip 10pt

{\it Example:} We now provide an example to help illustrate the
method used in the proof of Theorem \ref{non-exst-1}. Let
$n=13,r=\delta=2$ and $k=7$. Suppose $T_1=\{1,2,3\}$,
$T_2=\{4,5,6\}$, $T_3=\{7,8,9\}, T_4=\{10,11,12\}, T_5=\{1,5,13\}$
and $T_6=\{5,8,13\}$. Following the notations in the proof of
Theorem \ref{non-exst-1}, we have
\begin{center}
\includegraphics[height=3cm]{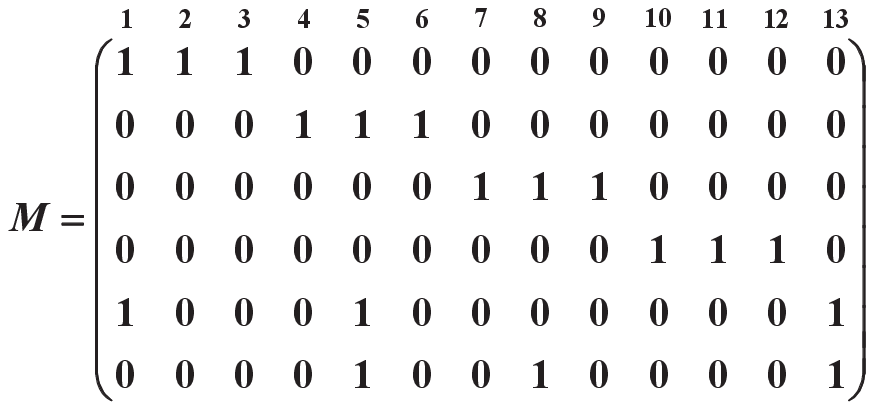}
\end{center}
Therefore, $A_1=\{1,5\}, A_5=\{2,5,6\}, A_8=\{3,6\},
A_{13}=\{5,6\}$, and $P=\{1,5,8,13\}$. Note that
$|A_1|+|A_5|=5>2+(r-v+1)$. Let $B_2=\{2,5\}\subseteq A_5$ and
$B=A_1\cup B_2=\{1,2,5\}$; then $|B|<4=\lceil\frac{k}{r}\rceil$.
Let $J=\{1,2,3,5\}\supseteq B$, then $\cup_{i\in
J}T_i=\{1,2,3,4,5,6,7,8,9,13\}.$  Hence, $|\cup_{i\in
J}T_i|=10<11=k+\lceil\frac{k}{r}\rceil(\delta-1)$. (See the
illustration of $M$ below.)
\begin{center}
\includegraphics[height=3cm]{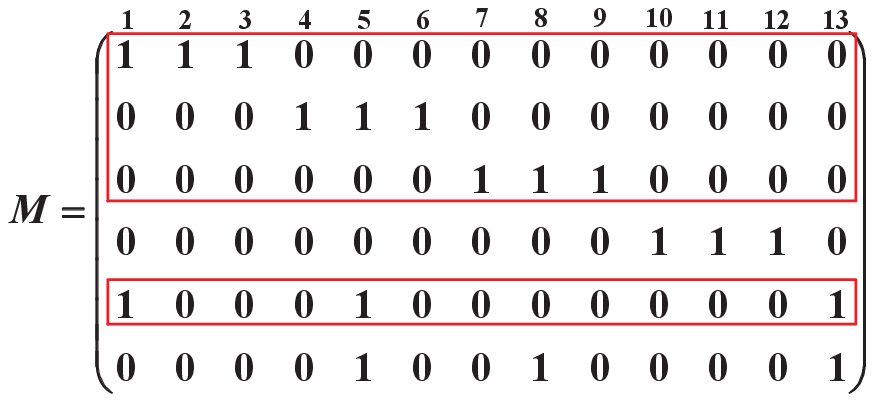}
\end{center}
More generally, in this example, for any $t\geq 5$ and
$\{T_1,\cdots,T_t\}$ such that $|T_i|=r+\delta-1=3$ and
$\cup_{i=1}^tT_i=[n]=\{1,\cdots,13\}$, we can always find a
$J\subseteq[t]$ such that $|\cup_{i\in
J}T_i|<11=k+\lceil\frac{k}{r}\rceil(\delta-1)$.

In general,  since $0<v<r$, then $r-v\leq r-1$. If $k>2r^2+r$,
then we have $u\geq 2(r-1)+1\geq 2(r-v)+1$. Hence, when
$0<n~\text{mod~}(r+\delta-1)<(k~\text{mod}~r)+\delta-1$ and
$k>2r^2+r$, then by Theorem \ref{non-exst-1}, there exist no
optimal $(r,\delta)_a$ codes.

\section{Construction of Optimal $(r,\delta)_a$ Codes: Algorithm 1}

In this section, we propose a deterministic algorithm for
constructing optimal $(r,\delta)_a$ linear codes over the field of
size $q\geq\binom{n}{k-1}$, when $(r+\delta-1)|n$ or $m\geq
v+\delta-1$, where $n=w(r+\delta-1)+m$ and $k=ur+v$ satisfying
$0<v<r$ and $0<m<r+\delta-1$. Recall that when $(r+\delta-1)|n$,
it was proved in \cite{Prakash12} that optimal $(r,\delta)_a$
linear codes exist over the field of size $q>kn^k$. Note that our
method requires a much smaller field than what's shown in
\cite{Prakash12}, and hence it also has a lower complexity for
implementation.

To present our method, we will use the following definitions and
notations, most of which follow from \cite{Gopalan12}.

\vskip 10pt
\begin{defn}\label{core}
Let $\mathcal S=\{S_1,\cdots, S_t\}$ be a partition of $[n]$ and
$\delta\leq|S_i|\leq r+\delta-1, \forall i\in[t]$. A subset
$S\subseteq[n]$ is called an $(\mathcal S,r)$-\emph{core} if
$|S\cap S_i|\leq|S_i|-\delta+1, \forall i\in[t]$. If $S$ is an
$(\mathcal S,r)$-core and $|S|=k$, then $S$ is called an
$(\mathcal S,r,k)$-\emph{core}.
\end{defn}
\vskip 10pt

Clearly, if $S\subseteq[n]$ is an $(\mathcal S,r)$-core and
$S'\subseteq S$, then $S'$ is also an $(\mathcal S,r)$-core. In
particular, if $S\subseteq[n]$ is an $(\mathcal S,r)$-core and
$S'$ is a $k$-subset of $S$, then $S'$ is an $(\mathcal
S,r,k)$-core.

Before presenting our construction method, we first give a lemma,
which will take an important role in our discussion.

\vskip 10pt
\begin{lem}\label{sub-space}
Let $X_1,\cdots,X_\ell$ and $X$ be $\ell+1$ subspaces of $\mathbb
F_q^k$ and $X\nsubseteq X_i, \forall i\in[\ell]$. If $q\geq\ell$,
then $X\nsubseteq\cup_{i=1}^\ell X_i$.
\end{lem}
\begin{proof}
We prove this lemma by induction.

Clearly, the claim is true when $\ell=1$.

Now, we suppose that the claim is true for $\ell-1$, i.e.,
$$X\nsubseteq\cup_{i=1}^{\ell-1} X_i.$$ Then there exists an
$x\in X$ such that $x\notin\cup_{i=1}^{\ell-1} X_i$. If $x\notin
X_\ell$, then $x\notin\cup_{i=1}^{\ell} X_i$ and
$X\nsubseteq\cup_{i=1}^\ell X_i$. So we assume $x\in X_\ell$.

Since $X\nsubseteq X_\ell$, there exists a $y\in X$ such that
$y\notin X_\ell$. Then for any $\{a,a'\}\subseteq\mathbb F_q$ and
$i\in\{1,\cdots,\ell-1\}$,
$$\{ax+y,a'x+y\}\nsubseteq X_i.$$ $($Otherwise,
$(a-a')x=(ax+y)-(a'x+y)\in X_i$, which contradicts to the
assumption that $x\notin\cup_{i=1}^{\ell-1} X_i.)$

Since $q\geq\ell$, we can pick a subset
$\{a_1,\cdots,a_\ell\}\subseteq\mathbb F_q$. Then $\{a_1x+y,
\cdots, a_\ell x+y\}\nsubseteq\cup_{i=1}^{\ell-1} X_i.
~($Otherwise, by the Pigeonhole principle, there is a subset
$\{a_{i_1},a_{i_2}\}\subseteq\{a_1,\cdots,a_\ell\}$ and a
$j\in\{1,\cdots,\ell-1\}$ such that
$\{a_{i_1}x+y,a_{i_2}x+y\}\subseteq X_j$, which contradicts to the
proven result that for any $\{a,a'\}\subseteq\mathbb F_q$ and
$i\in\{1,\cdots,\ell-1\}$, $\{ax+y,a'x+y\}\nsubseteq X_i.)$
Without loss of generality, assume
$a_1x+y\notin\cup_{i=1}^{\ell-1} X_i$. Note that $x\in X_\ell$ and
$y\notin X_\ell$, then $a_1x+y\notin X_\ell$. Hence,
$a_1x+y\notin\cup_{i=1}^{\ell} X_i$. On the other hand, since
$x,y\in X$, then $a_1x+y\in X$. So $X\nsubseteq\cup_{i=1}^\ell
X_i$, which completes the proof.
\end{proof}

\vskip 10pt

We present our construction method in the following theorem.

\vskip 10pt
\begin{thm}\label{opt-suf-1}
Let $\mathcal S=\{S_1,\cdots,S_t\}$ be a partition of $[n]$ and
$\delta\leq|S_i|\leq r+\delta-1, \forall i\in[t]$. Suppose
$t\geq\lceil\frac{k}{r}\rceil$ and for any
$\lceil\frac{k}{r}\rceil$-subset $J$ of $[t]$, $|\cup_{i\in
J}S_i|\geq k+\lceil\frac{k}{r}\rceil(\delta-1)$. If
$q\geq\binom{n}{k-1}$, then there exists an optimal $(r,\delta)_a$
linear code over $\mathbb F_q$.
\end{thm}
\begin{proof}
For each $i\in[t]$, let $U_i$ be an $(|S_i|-\delta+1)$-subset of
$S_i$. Let $\Omega_0=\cup_{i\in[t]}U_i$ and $L=|\Omega_0|$. Let
$J$ be a $\lceil\frac{k}{r}\rceil$-subset of $[t]$. Since
$\cup_{i\in J}U_i\subseteq\Omega_0$, from the assumptions of this
theorem,
$$L=|\Omega_0|\geq|\cup_{i\in J}U_i|=|\cup_{i\in
J}S_i|-\lceil\frac{k}{r}\rceil(\delta-1)\geq k.$$

The construction of an optimal $(r,\delta)_a$ code consists of the
following two steps:

\emph{Step 1}: Construct an $[L,k]$ MDS code $\mathcal C_0$ over
$\mathbb F_q$. Since $q\geq\binom{n}{k-1}\geq n>L$, such an MDS
code exists over $\mathbb F_q$. Let $G'$ be a generating matrix of
$\mathcal C_0$. We index the columns of $G'$ by $\Omega_0$, i.e.,
$G'=(G_\ell)_{\ell\in\Omega_0}$, where $G_\ell$ is a column of
$G'$ for each $\ell\in\Omega_0$.

\emph{Step 2}: Extend $\mathcal C_0$ to an optimal $(r,\delta)_a$
code $\mathcal C$ over $\mathbb F_q$. This can be achieved by the
following algorithm.

\vspace{0.12in} \noindent \textbf{Algorithm 1:}

\noindent 1. ~ Let $\Omega=\Omega_0$.

\noindent 2. ~ $i$ runs from $1$ to $t$.

\noindent 3. ~ ~ While $S_{i}\backslash\Omega\neq\emptyset$:

\noindent 4. ~ ~ ~ ~Pick a $\lambda\in S_{i}\backslash\Omega$ and
let $G_\lambda\in\langle\{G_\ell;~\ell\in S_i\cap\Omega\}\rangle$

~ ~ ~ ~be such that for any $(\mathcal S,r,k)$-core
$S~\subseteq\Omega~\cup~\{\lambda\}$,

~ ~ ~ ~$\{G_\ell; ~\ell\in S\}$ is linearly independent.

\noindent 5. ~ ~ ~ ~$\Omega=\Omega\cup\{\lambda\}$.

\noindent 6. ~ Let $\mathcal C$ be the linear code generated by
the matrix $~G=$

~ $(G_1,\cdots,G_n)$.
\vskip 10pt
%In the sequel, we will prove that if
%$q\geq\binom{n}{k-1}$, then in line 4 of Algorithm 1, we can
%always find a $G_\lambda$ satisfying the requirement. Hence, by
%induction, for any $(\mathcal S,r,k)$-core $S\subseteq[n]$,
%$\{G_\ell;\ell\in S\}$ is linearly independent.

To complete the proof of Theorem \ref{opt-suf-1}, we need to prove
three claims: In Claim~1 and Claim~2 below we show that the code
$\mathcal{C}$ output by Algorithm~1 is indeed an optimal
$(r,\delta)_a$ linear code over $\mathbb F_q$; In Claim~3, we
prove that the vector $G_\lambda$ described in Line 4 of
Algorithm~1 can always be found, hence the algorithm does
terminate successfully. \\

\noindent\textbf{Claim 1:} The code $\mathcal{C}$ output by
Algorithm~1 is an $(r,\delta)_a$ linear code over $\mathbb F_q$.

By Definition \ref{def-locality} and Remark \ref{rem-loty}, we aim
to show that for every $i \in [t]$ and for every subset $I \subset
S_i$ with $|I|=|S_i|-\delta+1$, it holds that
\begin{equation}
\label{eq:h-1} \text{Rank}(\{G_\ell\}_{\ell\in
I})=\text{Rank}(\{G_\ell\}_{\ell\in S_i}).
\end{equation}
Since in Line 4 of Algorithm 1, we choose
$G_\lambda\in\langle\{G_\ell;\ell\in S_i\cap\Omega\}\rangle$, we
have
$$\text{Rank}(\{G_\ell\}_{\ell\in(S_i\cap\Omega)\cup\{\lambda\}})
=\text{Rank}(\{G_\ell\}_{\ell\in S_i\cap\Omega}).$$ By induction,
\begin{equation}
\label{eq:h-2}
\begin{split}
\text{Rank}(\{G_\ell\}_{\ell\in
S_i})&= \text{Rank}(\{G_\ell\}_{\ell\in S_i\cap\Omega_0})\\
&= \text{Rank}(\{G_\ell\}_{\ell\in U_i})\\
&= |S_i|-\delta+1.
\end{split}
\end{equation}
%%%%%%%%%%%%%%%%%%%%%%%
Suppose $i\in[t]$ and $I\subseteq S_i$ such that
$|I|=|S_i|-\delta+1$. Then $|I|=|S_i|-\delta+1\leq r\leq k$. Since
$t\geq\lceil\frac{k}{r}\rceil$, we can find a
$\lceil\frac{k}{r}\rceil$-subset $J'$ of $[t]$ such that $i\in
J'$. For each $j\in J'$, let $W_j$ be an $(|S_j|-\delta+1)$-subset
of $S_j$ such that $W_i=I$. Clearly, $\cup_{j\in J'}W_j$ is an
$(\mathcal S,r)$-core. From the assumption of this lemma,
$$|\cup_{j\in J'}S_j|\geq k+\lceil\frac{k}{r}\rceil(\delta-1).$$ Hence
$$|\cup_{j\in J'}W_j|=|\cup_{j\in J'}S_j|-
\lceil\frac{k}{r}\rceil(\delta-1)\geq k.$$ Let $S$ be a $k$-subset
of $\cup_{j\in J'}W_j$ such that $I\subseteq S$, then $S$ is an
$(\mathcal S,r,k)$-core. Therefore, $\{G_\ell;\ell\in S\}$ is
linearly independent, which in turn implies that $\{G_\ell;\ell\in
I\}$ is also linearly independent. Therefore,
\begin{equation}
\label{eq:h-3} \text{Rank}(\{G_\ell\}_{\ell\in I}) = |I| =
|S_i|-\delta+1.
\end{equation}
Combining (\ref{eq:h-2}) and (\ref{eq:h-3}) we obtain (\ref{eq:h-1}). \\

\noindent\textbf{Claim 2:} The code $\mathcal{C}$ output by
Algorithm~1 has minimum distance achieving the upper bound
(\ref{eqn:1}), and hence is an optimal $(r,\delta)_a$ linear code.

According to Lemma~\ref{fact} and (\ref{eqn:1}), it suffices to
prove that for any subset $T\subseteq[n]$ of size
$|T|=k+(\lceil\frac{k}{r}\rceil-1)(\delta-1)$,
\[
\text{Rank}(\{G_\ell\}_{\ell\in T}) = k.
\]
Let
$$J=\{j\in[t];|T\cap S_j|\geq|S_j|-\delta+1\}.$$ For each
$j\in J$, let $W_j$ be an $(|S_j|-\delta+1)$-subset of $T\cap
S_j$; For each $j\in[t]\backslash J$, let $W_j=T\cap S_j$. Then
$\cup_{j\in[t]}W_j$ is an $(\mathcal S,r)$-core. We consider the
following two cases:

Case 1: $|J|\geq\lceil\frac{k}{r}\rceil$. Without loss of
generality, assume that $|J|=\lceil\frac{k}{r}\rceil$\footnote{If
$|J|>\lceil\frac{k}{r}\rceil$, then pick a
$\lceil\frac{k}{r}\rceil$-subset $J_0$ of $J$, and replace $J$ by
$J_0$ in our discussion.}. Since $|\cup_{j\in J}S_j|\geq
k+\lceil\frac{k}{r}\rceil(\delta-1)$, then
$$|\cup_{j\in[t]}W_j|\geq|\cup_{j\in J}W_j|\geq k.$$

Case 2: $|J|\leq\lceil\frac{k}{r}\rceil-1$. In that case,
 $$|\cup_{j\in
[t]}W_j|\geq |T|-|J|(\delta-1)\geq
|T|-(\lceil\frac{k}{r}\rceil-1)(\delta-1)\geq k.$$

In both cases, $|\cup_{j\in [t]}W_j|\geq k$. Let $S$ be a
$k$-subset of $\cup_{j\in J}W_j$, then $S$ is an $(\mathcal
S,r,k)$-core. Therefore, $\{G_\ell;\ell\in S\}$ are linearly
independent and
$$\text{Rank}(\{G_\ell\}_{\ell\in
T})=\text{Rank}(\{G_\ell\}_{\ell\in S})=k.$$ From equation
(\ref{eqn:1}) and Lemma \ref{fact}, we get
$$d=n-k+1-(\lceil\frac{k}{r}\rceil-1)(\delta-1),$$ where $d$ is
the minimum distance of $\mathcal C$. Thus, $\mathcal C$ is an
optimal $(r,\delta)_a$ code.\\

\noindent\textbf{Claim 3:} The vector $G_\lambda$ in Line 4 of
Algorithm 1 can always be found.

The proof of this claim is based on a classical technique in
network coding $($e.g., \cite{Li03,Jaggi05}$)$. Since
$G'=(G_\ell)_{\ell\in\Omega_0}$ is a generating matrix of the MDS
code $\mathcal C_0$, then for any $(\mathcal S,r,k)$-core
$S\subseteq\Omega_0$, $\{G_\ell;\ell\in S\}$ is linearly
independent. By induction, we can assume that for any $(\mathcal
S,r,k)$-core $S\subseteq\Omega$, $\{G_\ell;\ell\in S\}$ are
linearly independent.

Let $\Lambda$ be the set of all $S_0\subseteq\Omega$ such that
$S_0\cup\{\lambda\}$ is an $(\mathcal S,r,k)$-core. By Definition
\ref{core}, for any $S_0\in\Lambda$, $$|S_0|=k-1,$$
$$|S_0\cap S_j|\leq|S_j|-\delta+1, \ \forall j\in[t]\backslash\{i\},$$ and
$$|S_0\cap S_i|\leq|S_i|-\delta.$$ Note that
$$U_i\subseteq S_i\cap\Omega_0\subseteq S_i\cap\Omega.$$
Hence
$$|S_i\cap\Omega|\geq|U_i|=|S_i|-\delta+1.$$ Thus, there is an
$\eta\in(S_i\cap\Omega)\backslash S_0$. Since $S_1,\cdots,S_{t}$
are mutually disjoint, $\eta\notin S_j, \forall
j\in[t]\backslash\{i\}$. Therefore,
$$|(S_0\cup\{\eta\})\cap S_j|\leq|S_j|-\delta+1, j=1,\cdots, t.$$ Then
$S_0\cup\{\eta\}\subseteq\Omega$ is an $(\mathcal S,r,k)$-core. By
assumption, $\{G_\ell\}_{\ell\in S_0\cup\{\eta\}}$ is linearly
independent. Hence
$$G_\eta\notin\langle\{G_\ell\}_{\ell\in S_0}\rangle, $$ and
$$\langle\{G_\ell\}_{\ell\in
S_i\cap\Omega}\rangle\nsubseteq\langle\{G_\ell\}_{\ell\in
S_0}\rangle.$$ Since $q\geq\binom{n}{k-1}\geq |\Lambda|$, by Lemma
\ref{sub-space},
$$\langle\{G_\ell\}_{\ell\in
S_i\cap\Omega}\rangle\nsubseteq(\cup_{S_0\in\Lambda}\langle\{G_\ell\}_{\ell\in
S_0}\rangle).$$ Let $G_\lambda$ be a vector in
$\langle\{G_\ell\}_{\ell\in
S_i\cap\Omega}\rangle\backslash(\cup_{S_0\in\Lambda}\langle\{G_\ell\}_{\ell\in
S_0}\rangle)$. Then for any $S_0\in\Lambda$, $\{G_\ell\}_{\ell\in
S_0\cup\{\lambda\}}$ are linearly independent.

Suppose $S\subseteq\Omega\cup\{\lambda\}$ is an $(\mathcal
S,r,k)$-core. If $\lambda\notin S$, then $S\subseteq\Omega$ and by
assumption, $\{G_\ell;\ell\in S\}$ is linearly independent. If
$\lambda\in S$, then $S_0=S\backslash\{\lambda\}\in\Lambda$ and by
the selection of $G_\lambda$, $\{G_\ell;\ell\in S\}$ is linearly
independent. Hence we always have that $\{G_\ell;\ell\in S\}$ is
linearly independent. Thus, the vector $G_\lambda$ satisfies the
requirement of Algorithm 1.
\end{proof}
\vskip 10pt

From the proof of Theorem \ref{opt-suf-1}, we can see that
$\mathcal S=\{S_1,\cdots,S_t\}$ is in fact an $(r,\delta)$-cover
set of the code $\mathcal C$, where $\mathcal C$ is the output of
Algorithm 1. The following example demonstrates how does Algorithm
1 work.

{\it Example:} We now construct an optimal $(r,\delta)_a$ linear
code with $r=\delta=2,k=3$ and $n=6$. Let
$S_1=\{1,2,3\},S_2=\{4,5,6\}$ and $\mathcal S=\{S_1,S_2\}$. Let
$U_1=\{1,2\}, U_2=\{4,5\}$ and $\Omega_0=U_1\cup U_2=\{1,2,4,5\}$.
Our construct involves the following two steps.

Step 1: Construct a $[4,3]$ MDS code, where $4=|\Omega_0|$. Let
$G'=(G_1,G_2,G_4,G_5)$ be a generating matrix of such code.

Step 2: Extend $G'=(G_1,G_2,G_4,G_5)$ to a matrix
$G=(G_1,G_2,G_3,G_4,G_5,G_6)$ such that $G$ is a generating matrix
of an optimal $(2,2)_a$ linear code.

It remains to determine $G_3$ and $G_6$ via two iterations.
\begin{enumerate}
    \item $i = 1$: $\Omega = \{1,2,4,5\}$ and $S_1 \setminus \Omega = \{3\}$.
We can verify that $\{1,4,3\},\{1,5,3\}$, $\{2,4,3\}$, $\{2,5,3\}$
and $\{4,5,3\}$ are all subsets of $\{1,2,3,4,5\}$ which is an
$(\mathcal S,r,k)$-core and contains the index $3$. Let
$\Lambda=\{\{1,4\},\{1,5\}$, $\{2,4\}$, $\{2,5\},\{4,5\}\}$. Since
$G'=(G_1,G_2,G_4,G_5)$ generates an MDS code, then $G_1,G_2$ and
$G_4$ are linearly independent. So $\langle
G_1,G_2\rangle\nsubseteq\langle G_1,G_4\rangle$. Similarly,
$\langle G_1,G_2\rangle\nsubseteq\langle G_i,G_j\rangle, \forall
\{i,j\}\in\Lambda$. By Lemma \ref{sub-space}, if
$q\geq|\Lambda|=5$, then $\langle
G_1,G_2\rangle\nsubseteq\cup_{\{i,j\}\in\Lambda}\langle
G_i,G_j\rangle$. Note that $S_1 \cap \Omega = \{1,2\}$. Therefore,
let
$$G_3\in\langle G_1,G_2\rangle\backslash(\cup_{\{i,j\}\in\Lambda}\langle
G_i,G_j\rangle).$$ Then for any $(\mathcal S,r,k)$-core
$S\subseteq\{1,2,3,4,5\}$, $\{G_\ell; \ell\in S\}$ is linearly
independent.
    \item $i = 2$: $\Omega = \{1,2,3,4,5\}$ and $S_2 \setminus \Omega = \{6\}$.
Similarly, we can verify that $\{1,2,6\},\{1,3,6\}$, $\{2,3,6\}$,
$\{1,4,6\},\{1,5,6\}$, $\{2,4,6\}$ and $\{2,5,6\}$ are all subsets
which is an $(\mathcal S,r,k)$-core and contains the index $6$.
Let $\Lambda=\{\{1,2\}$, $\{1,3\}$, $\{2,3\}, \{1,4\},\{1,5\}$,
$\{2,4\}$, $\{2,5\}\}$. Clearly, $\langle
G_4,G_5\rangle\nsubseteq\langle G_i,G_j\rangle, \forall
\{i,j\}\in\Lambda$. By Lemma \ref{sub-space}, if
$q\geq|\Lambda|=7$, then $\langle
G_4,G_5\rangle\nsubseteq\cup_{\{i,j\}\in\Lambda}\langle
G_i,G_j\rangle$. As $S_2 \cap \Omega = \{4,5\}$, let
$$G_6\in\langle G_4,G_5\rangle\backslash(\cup_{\{i,j\}\in\Lambda}\langle
G_i,G_j\rangle).$$ Then for any $(\mathcal S,r,k)$-core $S$,
$\{G_\ell; \ell\in S\}$ is linearly independent. Thus, we can
obtain a matrix $G=(G_1,G_2,G_3,G_4,G_5,G_6)$ such that  for any
$(\mathcal S,r,k)$-core $S$, $\{G_\ell; \ell\in S\}$ is linearly
independent. Let $\mathcal C$ be the linear code generated by $G$.
Then $\mathcal C$ is an optimal $(2,2)_a$ linear code.
\end{enumerate}
We can in fact employ a smaller field than $\mathbb{F}_7$. The
following is a generating matrix of an optimal $(2,2)_a$ linear
code:
\begin{equation*}
G=\left(\begin{array}{cccccc}
1 & 0 & 1 & 0 & 1 & 1\\
0 & 1 & 1 & 0 & \alpha & \alpha\\
0 & 0 & 0 & 1 & 1 & \alpha\\
\end{array}\right)
\end{equation*}
over the field $\mathbb F_4=\{0,1,\alpha,1+\alpha\}$, where
$\alpha^2=1+\alpha$. \vskip 10pt

In the rest of this section, we shall use Theorem \ref{opt-suf-1}
to prove that optimal $(r,\delta)_a$ linear codes exist over a
field of size $q\geq\binom{n}{k-1}$ when $(r+\delta-1)|n$ or
$m\geq v+\delta-1$, where $n=w(r+\delta-1)+m$ and $k=ur+v$
satisfying $0<v<r$ and $0<m<r+\delta-1$. By Claim 2) of Lemma
\ref{low-bound}, $\frac{n}{r+\delta-1}\geq\frac{k}{r}$ is a
necessary condition for the existence of optimal $(r,\delta)_a$
linear codes. For this reason, we assume
$\frac{n}{r+\delta-1}\geq\frac{k}{r}$ holds in both cases.

\vskip 10pt
\begin{thm}\label{opt-ext-1}
Suppose $(r+\delta-1)|n$. % and $d\geq\delta$.
If $q\geq\binom{n}{k-1}$, then there exists an optimal
$(r,\delta)_a$ linear code over $\mathbb F_q$.
\end{thm}
\begin{proof}
Let $n=t(r+\delta-1)$. %Since $\delta\leq d$, by (1), we have
%\begin{eqnarray*}
%nr&=&r(k+d-1)+r(\lceil\frac{k}{r}\rceil-1)(\delta-1)\\
%&\geq&r(k+\delta-1)+r(\lceil\frac{k}{r}\rceil-1)(\delta-1)\\
%&\geq&r(k+\delta-1)+r(\frac{k}{r}-1)(\delta-1)\\
%&=&k(r+\delta-1).
%\end{eqnarray*}
%So $$t=\frac{n}{r+\delta-1}\geq\lceil\frac{k}{r}\rceil.$$ Hence,
%$$n-t(\delta-1)=rt\geq r\lceil\frac{k}{r}\rceil\geq r\frac{k}{r}=k.$$
Note that we have assumed that
$\frac{n}{r+\delta-1}\geq\frac{k}{r}$. Then
$$t=\lceil\frac{n}{r+\delta-1}\rceil\geq\lceil\frac{k}{r}\rceil.$$
Let $\{S_1,\cdots,S_t\}$ be a partition of $\{1,\cdots,n\}$ such
that $|S_i|=r+\delta-1, i=1,\cdots,t$.

For any $J\subseteq[t]$ of size $|J|=\lceil\frac{k}{r}\rceil$,
$$|\cup_{i\in J}S_i|=\lceil\frac{k}{r}\rceil(r+\delta-1)\geq
k+\lceil\frac{k}{r}\rceil(\delta-1).$$ By Theorem \ref{opt-suf-1},
if $q\geq\binom{n}{k-1}$, then there exists an optimal
$(r,\delta)_a$ code over $\mathbb F_q$.
\end{proof}
\vskip 10pt

If $(r+\delta-1)|n$ and $\delta\leq d$, then following a similar
line of proof in \cite{Prakash12}, we can show that
$t=\lceil\frac{n}{r+\delta-1}\rceil\geq\lceil\frac{k}{r}\rceil$.
Under these two conditions, it was proved in \cite{Prakash12} that
there exists an optimal $(r,\delta)_a$ code over the field
$\mathbb F_q$ of size $q>kn^k$. Our method requires a field of
size only $\binom{n}{k-1}$, which is at the largest a fraction
$\frac{1}{k!}$ of $kn^k$.

\vskip 10pt
\begin{thm}\label{opt-ext-2}
Suppose $n=w(r+\delta-1)+m$ and $k=ur+v$, where $0<m<r+\delta-1$
and $0<v<r$. Suppose $m\geq v+\delta-1$ and $d\geq\delta$. If
$q\geq\binom{n}{k-1}$, then there exists an optimal $(r,\delta)_a$
linear code over $\mathbb F_q$.
\end{thm}
\begin{proof}
Let $t=w+1$. Since we have assumed that
$\frac{n}{r+\delta-1}\geq\frac{k}{r}$, we get
$$t=w+1=\lceil\frac{n}{r+\delta-1}\rceil\geq\lceil\frac{k}{r}\rceil= u+1.$$
Note that
$n-m=w(r+\delta-1)$. Let $\{S_1,\cdots,S_w\}$ be a partition of
$\{1,\cdots,n-m\}$ and $S_t=[n-m+1,n]$.

For any $J\subseteq[t]$ of size $|J|=\lceil\frac{k}{r}\rceil$, we
have the following two cases:

Case 1: $t\notin J$. Then
$$|\cup_{i\in J}S_i|=\left\lceil\frac{k}{r}\right\rceil(r+\delta-1)\geq
k+\left\lceil\frac{k}{r}\right\rceil(\delta-1).$$

Case 2: $t\in J$. Since $m\geq v+\delta-1$, then
\begin{eqnarray*}
|\cup_{i\in J}S_i|&=&(\left\lceil\frac{k}{r}\right\rceil-1)(r+\delta-1)+m,\\
&\geq&(\left\lceil\frac{k}{r}\right\rceil-1)(r+\delta-1)+v+\delta-1,\\
&=&k+\left\lceil\frac{k}{r}\right\rceil(\delta-1).
\end{eqnarray*}

Hence, for any $\lceil\frac{k}{r}\rceil$-subset $J$ of $[t]$,
$|\cup_{i\in J}S_i|\geq k+\lceil\frac{k}{r}\rceil(\delta-1)$. By
Theorem \ref{opt-suf-1}, if $q\geq\binom{n}{k-1}$,  there
exists an optimal $(r,\delta)_a$ code over $\mathbb F_q$.
\end{proof}
\vskip 10pt

When $\delta=2$, the conditions of Theorem \ref{opt-ext-1} and
Theorem \ref{opt-ext-2} become $(r+1)|n$ and
$n~\text{mod}~(r+1)-1\geq k~\text{mod}~r>0$ respectively. For this
special case, Tamo \emph{et al.} \cite{Tamo13} introduced a
different construction method which is very easy to implement. However,
the method in \cite{Tamo13} requires the field size $q=O(n^k)$,
which is larger than the field size $q=\binom{n}{k-1}$ of our method.

\section{Construction of Optimal $(r,\delta)_a$ Codes: Algorithm 2}
In this section, we present yet another method for constructing
optimal $(r,\delta)_a$ codes. This constructive method also points
out two other sets of coding parameters where optimal
$(r,\delta)_a$ codes exist. As the method in Section
\uppercase\expandafter{\romannumeral5}, this method construct an
optimal $(r,\delta)_a$ code which has a given set $\mathcal S$ as
its $(r,\delta)$-cover set. The difference is that the set
$\mathcal S$ used by this method has a more complicated structure.
We again borrow the notion of \emph{core} from \cite{Gopalan12}.

\vskip 10pt
\begin{defn}\label{frame}
Let $\mathcal S=\{S_1,\cdots,S_t\}$ be a collection of
$(r+\delta-1)$-subsets of $[n]$, $\mathcal
A=\{A_1,\cdots,A_\alpha,B\}$ be a partition of $[t]$ and
$\Psi=\{\xi_1,\cdots,\xi_\alpha\}\subseteq[n]$. We say that
$\mathcal S$ is an $(\mathcal A,\Psi)$-\emph{frame} over the set
$[n]$, if the following two conditions are satisfied:
\begin{itemize}
  \item [(1)] For each $j\in[\alpha]$, $\{\xi_j\}=\cap_{\ell\in A_j}S_\ell$
  and $\{S_i\backslash\{\xi_j\}; i\in A_j\}$ are mutually disjoint;
%  \item [(1)] For each $j\in[\alpha]$, $\{\xi_j\}=\cap_{\ell\in A_j}S_\ell$
%  and the collection $\{S_i\backslash\{\xi_j\}; i\in A_j\}$ is a partition
%  of $(\cup_{\ell\in A_j}S_\ell)\backslash\{\xi_j\}$;
  \item [(2)] $\{\cup_{\ell\in A_j}S_\ell;j\in[\alpha]\}\cup\{S_j;j\in B\}$
  is a partition of $[n]$.
\end{itemize}
\end{defn}
\vskip 10pt

\begin{exam}\label{eg-core}
Let $\mathcal S=\{S_1,\cdots,S_8\}$ be what's shown in Fig \ref{fig-core}.
Clearly $\mathcal S$ is an $(\mathcal A,\Psi)$-frame over $[n]$,
where the subsets $S_1,S_2,S_3$ have a common element $\xi_1=1$,
and the subsets $S_4,S_5$ have a common element $\xi_2=14$.
\end{exam}

%%%%%%%%%%%%%%%%%%%%%%%%%%%%%%%%%%%%%%%%%%%
\renewcommand\figurename{Fig}
\begin{figure}[htbp]
\begin{center}
\includegraphics[width=5cm]{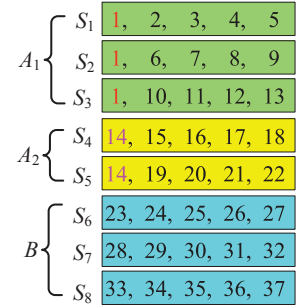}
\end{center}
\caption{An $(\mathcal A,\Psi)$-frame, where $n=37, r=\delta=3,
t=8, A_1=\{1,2,3\}$, $A_2=\{4,5\}, B=\{6,7,8\}, \mathcal A=\{A_1,
A_2, B\}$ and $\Psi=\{1,14\}$.}\label{fig-core}
\end{figure}
%%%%%%%%%%%%%%%%%%%%%%%%%%%%%%%%%%%%%%%%%%%%%%

\vskip 10pt
\begin{defn}\label{g-core}
A subset $S\subseteq[n]$ is said to be an $(\mathcal
S,r)$-\emph{core} if the following three conditions hold:
\begin{itemize}
  \item [(1)] If $j\in[\alpha]$ and $\xi_j\in S$, then
  $|S\cap S_i|\leq r, \forall i\in A_j$;
  \item [(2)] If $j\in[\alpha]$ and $\xi_j\notin S$, then
  there is an $i_j\in A_j$ such that $|S\cap S_{i_j}|\leq r$ and
  $|S\cap S_i|\leq r-1, \forall i\in A_j\backslash\{i_j\}$;
  \item [(3)] If $i\in B$, then $|S\cap S_i|\leq r$.
\end{itemize}
Additionally, if $S\subseteq[n]$ is an $(\mathcal S,r)$-core and $|S|=k$, then
$S$ is called an $(\mathcal S,r,k)$-\emph{core}.
\end{defn}
\vskip 10pt

Clearly, if $S\subseteq[n]$ is an $(\mathcal S,r)$-core and
$S'\subseteq S$, then $S'$ is also an $(\mathcal S,r)$-core. In
particular, if $S\subseteq[n]$ is an $(\mathcal S,r)$-core and
$S'$ is a $k$-subset of $S$, then $S'$ is an $(\mathcal
S,r,k)$-core.

{\it Example \ref{eg-core} continued:}
In Example \ref{eg-core}, let $k=7$. Then $\{1,2,3,6,7,10,11\}$
and $\{2,3,4,6,7,28,33\}$ are both $(\mathcal S,r,k)$-core.
However, $S=\{2,3,4,6,7,8,28\}$ and $S'=\{2,6,15,23,24,25,26\}$
are not $(\mathcal S,r)$-core, because $S$ does not satisfy Condition (2)
and $S'$ does not satisfy Condition (3) of
Definition \ref{g-core}.

%In general, if $T\subseteq[n]$ is not an $(\mathcal S,r)$-core,
%then we can find a subset $S\subseteq T$ such that $S$ is an
%$(\mathcal S,r)$-core, $|S|\geq $ and $S\geq $, where $J$ is the
%set of all $\ell\in[t]$ such that $|T\cap S_\ell|\geq r$.
\vskip 10pt
\begin{lem}\label{core-form}
Let $\mathcal S$ be an $(\mathcal A,\Psi)$-frame as in Definition
\ref{frame}. Suppose $t\geq\lceil\frac{k}{r}\rceil$ and for any
$\lceil\frac{k}{r}\rceil$-subset $J$ of $[t]$, $|\cup_{i\in J}S_i|
\geq k+\lceil\frac{k}{r}\rceil(\delta-1)$. Then the following
hold:
\begin{itemize}
  \item [1)] If $T\subseteq[n]$ has size $|T|\geq k+(\lceil\frac{k}{r}\rceil-1)
  (\delta-1)$, then there is an $S\subseteq T$ such that $S$ is an
  $(\mathcal S,r,k)$-core.
  \item [2)] For any $i\in[t]$ and $I\subseteq S_i$ of size $|I|=r$, there is an
  $(\mathcal S,r,k)$-core $S$ such that $I\subseteq S$.
\end{itemize}
\end{lem}
\begin{proof}
1) Let $$J=\{\ell\in[t]; |T\cap S_\ell|\geq r\}.$$

For each $j\in[\alpha]$ and $\ell\in A_j$, we pick a subset
$W_\ell\subseteq T$ as follows:

 \romannumeral1) If $J\cap
A_j=\emptyset$, then let $W_\ell=T\cap S_\ell$ for each $\ell\in
A_j$.

 \romannumeral2) If $J\cap A_j\neq\emptyset$ and $\xi_j\in
T$, then for each $\ell\in J\cap A_j$, let $W_\ell$ be an
$r$-subset of $T\cap S_\ell$ satisfying $\xi_j\in W_\ell$, and for
each $\ell\in A_j\backslash J$, let $W_\ell=T\cap S_\ell$.

\romannumeral3) If $J\cap A_j\neq\emptyset$ and $\xi_j\notin T$,
then fix an $\ell_j\in J\cap A_j$, and let $W_{\ell_j}$ be an
$r$-subset of $T\cap S_{\ell_j}$, let $W_\ell$ be an
$(r-1)$-subset of $T\cap S_\ell$ for each $\ell\in J\cap
A_j\backslash\{\ell_j\}$, and let $W_\ell=T\cap S_\ell$ for each
$\ell\in A_j\backslash J$.

Moreover, for each $\ell\in J\cap B$, let $W_\ell$ be an
$r$-subset of $T\cap S_\ell$, and for each $\ell\in B\backslash
J$, let $W_\ell=T\cap S_\ell$.

Let $W=\cup_{\ell\in[t]}W_\ell$, then by Definition \ref{g-core},
$W$ is an $(\mathcal S,r)$-core. We now prove that $|W|\geq k$.
Let
$$\Theta(J)=\{j\in[\alpha]; J\cap A_j\neq\emptyset\}.$$ %Since
%$\mathcal A=\{A_1,\cdots,A_\alpha,B\}$ is a partition of $[t]$,
%then
%$$J=(\cup_{j\in\Theta(J)}(J\cap A_j))\cup(J\cap B).$$
We need to consider the following two cases:

Case 1: $|J|\geq\lceil\frac{k}{r}\rceil$. Without loss of
generality, assume $|J|=\lceil\frac{k}{r}\rceil$\footnote{If
$|J|>\lceil\frac{k}{r}\rceil$, then pick a
$\lceil\frac{k}{r}\rceil$-subset $J_0$ of $J$, and replace $J$ by
$J_0$ in our discussion.}. Then from  the assumption of this lemma,
\begin{align}
|\cup_{\ell\in J}S_\ell|\geq k+|J|(\delta-1).\label{eqn:28}
\end{align}
By Definition \ref{frame},
\begin{align}
|\cup_{\ell\in J}S_\ell| = & \sum_{j\in\Theta(J)}|J\cap
A_j|(r+\delta-2) \nonumber\\
&\ \ +|\Theta(J)|+|J\cap B|(r+\delta-1).
\label{eqn:29}
\end{align}
Since $\mathcal A=\{A_1,\cdots,A_\alpha,B\}$ is a partition of
$[t]$, $\{J\cap A_j;j\in\Theta(J)\}\cup\{J\cap B\}$ is a partition
of $J$ and
\begin{align}
|J|=\sum_{j\in\Theta(J)}|J\cap A_j|+|J\cap
B|.\label{eqn:30}
\end{align}
Combining (\ref{eqn:28})$-$(\ref{eqn:30}), we have
\begin{align}
\sum_{j\in\Theta(J)}|J\cap A_j|(r-1)+|\Theta(J)|+|J\cap B|r\geq
k.\label{eqn:31}
\end{align}
By the construction of $W$, we have
\begin{align}
|\cup_{\ell\in J}W_\ell|=\sum_{j\in\Theta(J)}|J\cap A_j|(r-1)+
|\Theta(J)|+|J\cap B|r.\label{eqn:32}
\end{align}
Equations (\ref{eqn:31}) and (\ref{eqn:32}) imply that
$$|W|\geq|\cup_{\ell\in J}W_\ell|\geq k.$$

Case 2: $|J|<\lceil\frac{k}{r}\rceil$. By the construction of $W$,
for each $j\in[\alpha]$ and $\ell\in J\cap A_\ell$, $W_\ell$ is
obtained by deleting at most $(\delta-1)$ elements from $T\cap
S_\ell$. We thus have
$$|\cup_{\ell\in A_j}W_\ell|\geq|T\cap(\cup_{\ell\in
A_j}S_\ell)|-|J\cap A_j|(\delta-1).$$ Moreover,
$$|\cup_{\ell\in B}W_\ell|\geq|\cup_{\ell\in B}(T\cap S_\ell)|-|J\cap
B|(\delta-1).$$ Then $$|W|=|\cup_{\ell\in[t]}W_\ell|
\geq|T|-|J|(\delta-1).$$ Note that $|T|\geq
k+(\lceil\frac{k}{r}\rceil-1)(\delta-1)$ and
$|J|<\lceil\frac{k}{r}\rceil$. Therefore
$$|W|\geq|T|-(\lceil\frac{k}{r}\rceil-1)(\delta-1)=k.$$

Gathering both cases, we always have $|W|\geq k$. Let $S$ be a
$k$-subset of $W$. Note that $W$ is an $(\mathcal S,r)$-core. So
$S\subseteq W\subseteq T$ is an $(\mathcal S,r,k)$-core.

2) To prove the second claim of Lemma \ref{core-form}, note that
$t\geq\lceil\frac{k}{r}\rceil$, and hence we can always find a
$\lceil\frac{k}{r}\rceil$-subset $J$ of $[t]$ such that $i\in J$.
Similar to the proof of 1), for each $\ell\in J$, we can pick a
$W_\ell$ such that $W_i=I$, $\cup_{\ell\in J}W_\ell$ is an
$(\mathcal S,r)$-core and $|\cup_{\ell\in J}W_\ell|\geq k$. Let
$S$ be a $k$-subset of $\cup_{\ell\in J}W_\ell$ such that
$I\subseteq S$. Then $S$ is an $(\mathcal S,r,k)$-core and
$I\subseteq S$.
\end{proof}
\vskip 10pt

{\it Example \ref{eg-core} further continued}: Consider the
$(\mathcal A,\Psi)$-frame $\mathcal S$ in Example \ref{eg-core}.
Let $k=7$. Then $\mathcal S$ satisfies the conditions of Lemma
\ref{core-form}. We consider the following two instances:

Instance 1: $T=\{2,3,4,6,7,8,14,15,16,17,19,23,24,28\}$. As in the
proof of Lemma \ref{core-form}, $J=\{\ell;|T\cap S_\ell|\geq
r\}=\{1,2,4\}$ and $|J|=3=\lceil\frac{k}{r}\rceil$. Let
$W_1=\{2,3,4\}$, $W_2=\{6,7\}$, $W_4=\{14,15,16\}$, $W_5=\{19\}$,
$W_6=\{23,24\}$, $W_7=\{28\}$ and $W_\ell=\emptyset$ for
$\ell\in\{3,8\}$. Then
$|W|=|\cup_{\ell=1}^8W_\ell|\geq|\cup_{\ell\in J}W_\ell|\geq k=7$.

Instance 2: $T=\{2,3,4,6,7,8,10,11,14,15,19,23,24,28\}$. Then
$J=\{\ell;|T\cap S_\ell|\geq r\}=\{1,2\}$ and
$|J|<\lceil\frac{k}{r}\rceil$. Let $W_1=\{2,3,4\}, W_2=\{6,7\}$,
$W_3=\{10,11\}, W_4=\{14,15\}, W_5=\{19\},
W_6=\{23,24\},W_7=\{28\}$ and $W_8=\emptyset$. Then $|W|=$
$|\cup_{\ell=1}^8W_\ell|\geq|T|-|J|(\delta-1)\geq k=7$.

\vskip 10pt
\begin{rem}\label{nor-core}
Let $\mathcal S$ be an $(\mathcal A,\Psi)$-frame as in Definition
\ref{frame}. For each $j\in[\alpha]$ and $i\in A_j$, let $U_i$ be
an $r$-subset of $S_i$ such that $\xi_j\in U_i$. For each $i\in
B$, let $U_i$ be an $r$-subset of $S_i$. Let
$$\Omega_0=\cup_{i\in[t]}U_i.$$ Then
by Definition \ref{g-core}, $\Omega_0$ is an $(\mathcal
S,r)$-core. Clearly,
$$|\Omega_0|=n-t(\delta-1)=|\cup_{j=1}^\alpha A_j|(r-1)+\alpha+
|B|r.$$
\end{rem}
\vskip 10pt

\begin{exam}\label{eg-omg0}
In Example \ref{eg-core}, let $k=7$, then $\Omega_0=\{1,$
$2,3,6,7,10,11,14,15,16,19,20,23,24,25,28,29,30,33,34,$ $35\}$ is
an $(\mathcal S,r)$-core obtained by the process of Remark
\ref{nor-core}.
\end{exam}
\vskip 10pt

\begin{lem}\label{core-exist}
Let $\mathcal S$ be an $(\mathcal A,\Psi)$-frame as defined in Definition
\ref{frame} and $\Omega_0$ be what's described in Remark \ref{nor-core}. Suppose
$\Omega_0\subseteq\Omega\subseteq[n], S_0\subseteq\Omega$ and
$i\in[t]$. If $\lambda\in S_i\backslash\Omega$ and
$S_0\cup\{\lambda\}$ is an $(\mathcal S,r,k)$-core, then there exists
an $\eta\in S_i\cap\Omega$ such that $S_0\cup\{\eta\}$ is an
$(\mathcal S,r,k)$-core.
\end{lem}
\begin{proof}
By the construction of $\Omega_0$, $|S_i\cap\Omega_0|=r$. Since
$\Omega_0\subseteq\Omega$, so
$$|S_i\cap\Omega|\geq r.$$ Since $S_0\cup\{\lambda\}$ is
an $(\mathcal S,r,k)$-core, by Definition \ref{g-core},
$$|S_0|=k-1$$ and $$|S_0\cap S_i|\leq r-1.$$ Thus, we can find an
$\eta\in(S_i\cap\Omega)\backslash S_0$.

If $i\in B$, then by Definition \ref{frame}, $\eta\notin S_{i'},
\forall i'\in[t]\backslash\{i\}$. Then $S_0\cup\{\eta\}$ is an
$(\mathcal S,r,k)$-core.

Now, suppose $i\in A_j$ for some $j\in[\alpha]$. We need to
consider the following two cases.

Case 1: $\xi_j\in S_0$. Since $\eta\in(S_i\cap\Omega)\backslash
S_0$, then $\eta\neq\xi_j$ and $\eta\notin S_{i'}, \forall
i'\in[t]\backslash\{i\}$. Then $S_0\cup\{\eta\}$ is an $(\mathcal
S,r,k)$-core.

Case 2: $\xi_j\notin S_0$. Since $S_0\cup\{\lambda\}$ is an
$(\mathcal S,r,k)$-core, from Definition \ref{g-core}, we differentiate
the following two sub-cases:

Subcase 2.1: $|S_0\cap S_{i'}|\leq r-1, \forall i'\in A_j$. In that case, it is
clear that $S_0\cup\{\eta\}$ is an $(\mathcal S,r,k)$-core.

Subcase 2.2: There is an $i_j\in A_j\backslash\{i\}$ such that
$|S_0\cap S_{i_j}|=r$, $|S_0\cap S_{i}|\leq r-2$ and $|S_0\cap
S_{i'}|\leq r-1, \forall i'\in A_j\backslash\{i_j,i\}$. In that case, we have
$$|(S_i\cap\Omega)\backslash S_0|\geq 2.$$ Let
$\eta\in(S_i\cap\Omega)\backslash(S_0\cup\{\xi_j\})$, then
$\eta\neq\xi_j$ and $\eta\notin S_{i'}, \forall
i'\in[t]\backslash\{i\}$. It then follows that  $S_0\cup\{\eta\}$ is an
$(\mathcal S,r,k)$-core.
\end{proof}
\vskip 10pt

{\it Example \ref{eg-core} and \ref{eg-omg0} continued:}
Consider again Example \ref{eg-core}. Let $k=7$,
$\Omega=\Omega_0\cup\{4,5,8\}$ and $\lambda=9\in S_2$, where
$\Omega_0$ is as in Example \ref{eg-omg0}. We can easily verify the following:

Let $S_0=\{1,2,3,6$,
$10,14\}$; Then $S_0\cup\{9\}$ is an $(\mathcal S,r,k)$-core. If we further let
$\eta=7\in S_2$, then $S_0\cup\{\eta\}$ is also an $(\mathcal
S,r,k)$-core.

Let $S_0'=\{2,3,6,7,14,15\}$; Then $S_0'\cup\{9\}$
is an $(\mathcal S,r,k)$-core. If we further let $\eta'=8\in S_2$, then
$S_0'\cup\{\eta\}$ is also an $(\mathcal S,r,k)$-core.

Let
$S_0''=\{2,3,4,10,11,15,23\}$; Then $S_0''\cup\{9\}$ is an
$(\mathcal S,r,k)$-core. If we further let $\eta''=6\in S_2$, then
$S_0''\cup\{\eta''\}$ is also an $(\mathcal S,r,k)$-core.

\vskip 10pt
\begin{lem}\label{code-extd}
Let $\mathcal S$ be an $(\mathcal A,\Psi)$-frame defined in Definition
\ref{frame} and $\Omega_0$ be what's defined in Remark \ref{nor-core}. Let
$\Omega_0\subseteq\Omega\subseteq[n]$ and $\mathcal
G=\{G_\ell\in\mathbb F_q^k; \ell\in\Omega\}$ such that for any
$(\mathcal S,r,k)$-core $S\subseteq\Omega$, the vectors in
$\{G_\ell;\ell\in S\}$ are linearly independent. Suppose $i\in[t]$
and $S_i\backslash\Omega\neq\emptyset$. If $q\geq\binom{n}{k-1}$,
then for any $\lambda\in S_i\backslash\Omega$, there is a
$G_\lambda\in\langle\{G_\ell\}_{\ell\in S_i\cap\Omega}\rangle$
such that for any $(\mathcal S,r,k)$-core
$S\subseteq\Omega\cup\{\lambda\}$, the vectors in
$\{G_\ell;\ell\in S\}$ are linearly independent.
\end{lem}
\begin{proof}
Let $\Lambda$ be the set of all $S_0\subseteq\Omega$ such that
$S_0\cup\{\lambda\}$ is an $(\mathcal S,r,k)$-core. For any
$S_0\in\Lambda$, by Lemma \ref{core-exist}, there is an $\eta\in
S_i\cap\Omega$ such that $S_0\cup\{\eta\}$ is an $(\mathcal
S,r,k)$-core. From the assumptions, $\{G_\ell\}_{\ell\in
S_0\cup\{\eta\}}$ is linearly independent. Hence
$$G_\eta\notin\langle\{G_\ell\}_{\ell\in S_0}\rangle.$$ Thus,
$$\langle\{G_\ell\}_{\ell\in
S_i\cap\Omega}\rangle\nsubseteq\langle\{G_\ell\}_{\ell\in
S_0}\rangle.$$ Since $q\geq\binom{n}{k-1}\geq |\Lambda|$, by Lemma
\ref{sub-space},
$$\langle\{G_\ell\}_{\ell\in
S_i\cap\Omega}\rangle\nsubseteq(\cup_{S_0\in\Lambda}\langle\{G_\ell\}_{\ell\in
S_0}\rangle).$$ Let $G_\lambda\in\langle\{G_\ell\}_{\ell\in
S_i\cap\Omega}\rangle\backslash(\cup_{S_0\in\Lambda}\langle\{G_\ell\}_{\ell\in
S_0}\rangle)$. Then for any $(\mathcal S,r,k)$-core
$S\subseteq\Omega\cup\{\lambda\}$, the vectors in $\{G_\ell;\ell\in
S\}$ are linearly independent.
\end{proof}
\vskip 10pt

The second construction method for optimal $(r,\delta)_a$ codes
is illustrated in the proof of the following theorem.

\vskip 10pt
\begin{thm}\label{opt-suf}
Let $\mathcal S$ be an $(\mathcal A,\Psi)$-frame  in Definition
\ref{frame}. Suppose $t\geq\lceil\frac{k}{r}\rceil$ and for any
$\lceil\frac{k}{r}\rceil$-subset $J$ of $[t]$, $|\cup_{i\in J}S_i|
\geq k+\lceil\frac{k}{r}\rceil(\delta-1)$. If
$q\geq\binom{n}{k-1}$, then there exists an optimal $(r,\delta)_a$
linear code over $\mathbb F_q$.
\end{thm}
\begin{proof}
Let $\Omega_0$ be what's described in Remark \ref{nor-core} and $L=|\Omega_0|$.
Clearly, $$L=n-t(\delta-1).$$ Since
$t\geq\lceil\frac{k}{r}\rceil$, let $J$ be a
$\lceil\frac{k}{r}\rceil$-subset of $[t]$; then from the assumptions,
$$|\cup_{i\in J}S_i|\geq k+\lceil\frac{k}{r}\rceil(\delta-1)
=k+|J|(\delta-1).$$
By Remark \ref{nor-core}, $\cup_{i\in J}U_i\subseteq \Omega_0$. Hence
$$L=|\Omega_0|\geq|\cup_{i\in J}U_i|=\cup_{i\in J}S_i|-
|J|(\delta-1)\geq k.$$

The construction of an optimal $(r,\delta)_a$ code consists of the
following two steps.

\emph{Step 1}: Construct an $[L,k]$ MDS code $\mathcal C_0$ over
$\mathbb F_q$. Such an MDS code exists when
$q\geq\binom{n}{k-1}\geq n>L$. Let $G'$ be a generating matrix of
$\mathcal C_0$. We index the columns of $G'$ by $\Omega_0$, i.e.,
$G'=(G_\ell)_{\ell\in\Omega_0}$, where $G_\ell$ is a column of
$G', \forall\ell\in\Omega_0$.

\emph{Step 2}: Extend the code $\mathcal C_0$ to an optimal
$(r,\delta)_a$ code $\mathcal C$. This can be achieved by the
following algorithm, which appears similar to Algorithm 1 (on the surface)
but is actually different (in details).

\vspace{0.12in} \noindent \textbf{Algorithm 2:}

\noindent 1. ~ Let $\Omega=\Omega_0$.

\noindent 2. ~ $i$ runs from $1$ to $t$.

\noindent 3. ~ ~ While $S_{i}\backslash\Omega\neq\emptyset$:

\noindent 4. ~ ~ ~ ~Pick a $\lambda\in S_{i}\backslash\Omega$ and
let $G_\lambda\in\langle\{G_\ell;~\ell\in S_i\cap\Omega\}\rangle$

~ ~ ~ ~be such that for any $(\mathcal S,r,k)$-core
$S~\subseteq\Omega~\cup~\{\lambda\}$,

~ ~ ~ ~$\{G_\ell; ~\ell\in S\}$ is linearly independent.

\noindent 5. ~ ~ ~ ~$\Omega=\Omega\cup\{\lambda\}$.

\noindent 6. ~ Let $\mathcal C$ be the linear code generated by
the matrix $~G=$

~ $(G_1,\cdots,G_n)$.

\vspace{0.12in} Since $G'=(G_\ell)_{\ell\in\Omega_0}$ is a
generating matrix of the MDS code $\mathcal C_0$, so for any
$(\mathcal S,r,k)$-core $S\subseteq\Omega_0$, $\{G_\ell;\ell\in
S\}$ is linearly independent. Then in Algorithm 2, by induction,
we can assume that for any $(\mathcal S,r,k)$-core
$S\subseteq\Omega$, $\{G_\ell;\ell\in S\}$ is linearly
independent. By Lemma \ref{code-extd}, in line 4 of Algorithm 2,
we can always find a $G_\lambda$ satisfying the requirement.
Hence, by induction, the collection $\{G_\ell;\ell\in[n]\}$
satisfies the condition that for any $(\mathcal S,r,k)$-core
$S\subseteq[n]$, $\{G_\ell;\ell\in S\}$ is linearly independent.
Moreover, since in line 4 of Algorithm 2, we can choose a
$G_\lambda\in\langle\{G_\ell;\ell\in S_i\cap\Omega\}\rangle$,
which satisfies
$$\text{Rank}(\{G_\ell\}_{\ell\in(S_i\cap\Omega)\cup\{\lambda\}})
=\text{Rank}(\{G_\ell\}_{\ell\in S_i\cap\Omega}).$$ By induction,
\begin{align}
\text{Rank}(\{G_\ell\}_{\ell\in
S_i})&=\text{Rank}(\{G_\ell\}_{\ell\in
S_i\cup\Omega_0})\nonumber \\
&=\text{Rank}(\{G_\ell\}_{\ell\in U_i}) \nonumber\\
&=r.\nonumber
\end{align}
For any $i\in[t]$ and $I\subseteq S_i$ of size $|I|=r$, by Claim 2) of
Lemma \ref{core-form}, there is an $(\mathcal S,r,k)$-core $S$
such that $I\subseteq S$. Hence $\{G_\ell;\ell\in S\}$ is linearly
independent. Thus, $$\text{Rank}(\{G_\ell\}_{\ell\in I})=r.$$ Therefore,
by Definition \ref{def-locality} and Remark \ref{rem-loty},
$\mathcal C$ is an $(r,\delta)_a$ code.

Finally, we prove that the minimum distance of $\mathcal C$ is
$d=n-k+1-(\lceil\frac{k}{r}\rceil-1)(\delta-1)$.

Suppose $T\subseteq[n]$ and
$|T|=k+(\lceil\frac{k}{r}\rceil-1)(\delta-1)$. By 1) of Lemma
\ref{core-form}, there is an $S\subseteq T$ which is an $(\mathcal
S,r,k)$-core. Therefore, $$\text{Rank}(\{G_\ell; \ell\in
T\})=\text{Rank}(\{G_\ell; \ell\in S\})=k.$$ By the minimum distance bound in (\ref{eqn:1}) and
Lemma \ref{fact}, the minimum distance of $\mathcal C$ is
$$d=n-k+1-(\lceil\frac{k}{r}\rceil-1)(\delta-1).$$ Hence
$\mathcal C$ is an optimal $(r,\delta)_a$ code.
\end{proof}
\vskip 10pt

{\it Example \ref{eg-core} continued:}
Consider the $(\mathcal A,\Psi)$-frame $\mathcal S$ in Example
\ref{eg-core}. Let $k=7$. Then it is obvious $\mathcal S$ satisfies the
conditions of Theorem \ref{opt-suf}. Thus, we can use Algorithm 2 to
construct an optimal $(r,\delta)_a$ linear code over the field of
size $q\geq\binom{n}{k-1}=\binom{37}{6}$. Note that $r=\delta=3$.
Hence, $(r+\delta-1)\nmid n$ and this is a new optimal $(r,\delta)_a$
code.

\vskip 5pt
As applications of Theorem \ref{opt-suf}, in the following, we
show that  optimal $(r,\delta)_a$ codes exist for two other
sets of coding parameters. From Claim 2) of Lemma \ref{low-bound}, we know that
$\frac{n}{r+\delta-1}\geq\frac{k}{r}$ is a necessary condition for
the existence of optimal $(r,\delta)_a$ linear codes. Thus we will assume
$\frac{n}{r+\delta-1}\geq\frac{k}{r}$ in  the following discussion.

\vskip 10pt
\begin{thm}\label{opt-ext-3}
Suppose $n=w(r+\delta-1)+m$ and $k=ur+v$, where $0<m<r+\delta-1$
and $0<v<r$. Suppose $w\geq r+\delta-1-m$ and $r-v\geq u$. If
$q\geq\binom{n}{k-1}$, then there exists an optimal $(r,\delta)_a$
linear code over $\mathbb F_q$.
\end{thm}
\begin{proof}
Let $t=w+1$. Note that we have assumed that
$\frac{n}{r+\delta-1}\geq\frac{k}{r}$. Then
$$t=w+1=\lceil\frac{n}{r+\delta-1}\rceil\geq\lceil\frac{k}{r}\rceil=
u+1.$$ Let
$$\ell=r+\delta-1-m$$ and
\begin{align}
L=(\ell+1)(r+\delta-2)+1.\label{eqn:33}
\end{align} Then from the assumptions,
$w\geq(r+\delta-1)-m=\ell$. Therefore $$t=w+1\geq \ell+1$$ and
\begin{eqnarray}
n-L&=&(w-\ell)(r+\delta-1)\nonumber
\\&=&(t-\ell-1)(r+\delta-1).
\label{eqn:34}
\end{eqnarray}
From equation (\ref{eqn:33}), $L-1=(\ell+1)(r+\delta-2)$. The set $[2,L]$ can
be partitioned into $\ell+1$ mutually disjoint subsets, say,
$T_1,\cdots,T_{\ell+1}$, each of size $r+\delta-2$. Let
$$S_i=\{1\}\cup T_i, i=1,\cdots,\ell+1.$$
Moreover, from equation (\ref{eqn:34}), the set $[L+1,n]$ can be partitioned
into $t-(\ell+1)$ mutually disjoint subsets, say,
 $S_{\ell+2},\cdots,S_{t}$, each of size $r+\delta-1$.

Let $\alpha=1$ and $A_1=\{1,\cdots,\ell+1\},B=\{\ell+1,\cdots,t\}$,
$\mathcal A=\{A_1,B\}$,  and $\Psi=\{1\}$. Then $\mathcal
S=\{S_1,\cdots,S_t\}$ is an $(\mathcal A,\Psi)$-frame. For any
$\lceil\frac{k}{r}\rceil$-subset $J$ of $[t]$, since $r-v\geq u$,
then
$$|J|=\lceil\frac{k}{r}\rceil=u+1\leq r-v+1.$$ Let
$J_1=J\cap\{1,\cdots,\ell+1\},$ and
$J_2=J\backslash\{1,\cdots,\ell+1\}$. By the construction of
$\mathcal S$, we have
\begin{eqnarray*} |\cup_{i\in J}S_i|&=&|J_1|(r+\delta-2)+1+|J_2|(r+\delta-1)
\\&=&|J|(r+\delta-1)-|J_1|+1\\&\geq&|J|(r+\delta-1)-|J|+1
\\&\geq&|J|(r+\delta-1)-(r-v+1)+1\\&=&(|J|-1)r+v+|J|(\delta-1)
\\&=&ur+v+\lceil\frac{k}{r}\rceil(\delta-1)
\\&=&k+\lceil\frac{k}{r}\rceil(\delta-1).
\end{eqnarray*}
By Theorem \ref{opt-suf}, if $q\geq\binom{n}{k-1}$, then there
exists an optimal $(r,\delta)_a$ code over $\mathbb F_q$.
\end{proof}
\vskip 10pt

\begin{thm}\label{opt-ext-4}
Suppose $n=w(r+\delta-1)+m$ and $k=ur+v$, where $0<m<r+\delta-1$
and $0<v<r$. Suppose $w+1\geq 2(r+\delta-1-m)$ and $2(r-v)\geq u$.
If $q\geq\binom{n}{k-1}$, then there exists an optimal
$(r,\delta)_a$ linear code over $\mathbb F_q$.
\end{thm}
\begin{proof}
Let $t=w+1$. Note that we have assumed that
$\frac{n}{r+\delta-1}\geq\frac{k}{r}$. Then
$$t=w+1=\left\lceil\frac{n}{r+\delta-1}\right\rceil\geq
\left\lceil\frac{k}{r}\right\rceil=
u+1.$$ Let
$$\ell=(r+\delta-1)-m$$ and
\begin{align}
L=\ell(2(r+\delta-1)-1).\label{eqn:35}
\end{align}
Then by assumption, $t=w+1\geq 2(r+\delta-1-m)=2\ell$. It then follows that
\begin{align}
n-L=(t-2\ell)(r+\delta-1)\geq 0.\label{eqn:36}
\end{align}
From equation (\ref{eqn:35}), the set $[L]$ can be partitioned
into $\ell$ mutually disjoint subsets, say, $T_1,\cdots,T_{\ell}$,
each of size $2(r+\delta-1)-1$. For each $i\in\{1,\cdots,\ell\}$,
we can find two subsets $S_{2i-1},S_{2i}$ of $T_i$ such that
$$|S_{2i-1}|=|S_{2i}|=r+\delta-1$$ and $$S_{2i-1}\cup
S_{2i}=T_i.$$ Then
$$|S_{2i-1}\cap S_{2i}|=1.$$ Let $S_{2i-1}\cap
S_{2i}=\{\xi_i\}$ and $\Psi=\{\xi_1,\cdots,\xi_\ell\}$.

Moreover, from Equation (\ref{eqn:36}), the set $[L+1,n]$ can be partitioned
into $t-2\ell$ mutually disjoint subsets, say
$S_{2\ell+1},\cdots,S_{t}$, each of size $r+\delta-1$.

Let $A_i=\{2i-1,2i\}, i=1,\cdots,\ell$, $B=[2\ell+1,t]$ and
$\mathcal A=\{A_1,\cdots,A_\ell,B\}$. Then $\mathcal
S=\{S_1,\cdots,S_t\}$ is an $(\mathcal A,\Psi)$-frame. For any
$\lceil\frac{k}{r}\rceil$-subset $J$ of $[t]$. Since $2(r-v)\geq
u$, then
\begin{align}
|J|=\lceil\frac{k}{r}\rceil=u+1\leq 2(r-v)+1.
\label{eqn:37}
\end{align} Let
$\Gamma(J)=\{j\in[\ell];A_j\subseteq J\}$. Then
\begin{align}
|J|\geq|\cup_{j\in\Gamma(J)}A_j|=2|\Gamma(J)|.\label{eqn:38}
\end{align}
Combining (\ref{eqn:37}) an (\ref{eqn:38}), we have
$$|\Gamma(J)|\leq\frac{|J|}{2}\leq\frac{2(r-v)+1}{2}=r-v+\frac{1}{2}.$$
Since $|\Gamma(J)|$ is an integer, then $$|\Gamma(J)|\leq r-v.$$
By the construction of $\mathcal S$, we have
\begin{eqnarray*} |\cup_{i\in J}S_i|&=&|J|(r+\delta-1)-|\Gamma(J)|
\\&\geq&|J|(r+\delta-1)-(r-v)\\&=&(|J|-1)r+v+|J|(\delta-1)
\\&=&k+\lceil\frac{k}{r}\rceil(\delta-1).
\end{eqnarray*}
By Theorem \ref{opt-suf}, if $q\geq\binom{n}{k-1}$, then there
exists an optimal $(r,\delta)_a$ code over $\mathbb F_q$.
\end{proof}
\vskip 10pt

We now provide some discussions of Theorem \ref{opt-ext-4}. Since
$0<m<r+\delta-1$, then $2(r+\delta-1-m)<2(r+\delta-1)$. Given
$k,r$ and $\delta$, let
$\alpha=\text{max}\{2(r+\delta-1),\lceil\frac{k}{r}\rceil\}$. Then
the conditions $w+1\geq 2(r+\delta-1-m)$ and $w\geq u$ can always
be satisfied when $n\geq\alpha(r+\delta-1)$. On the other hand,
when $\frac{k}{3}<r<k$ and $r\neq\frac{k}{2}$, then $u=1$ or $2$
and $r-v\geq 1$, which leads to $2(r-v)\geq u$. By Theorem \ref{opt-ext-4},
there exist optimal $(r,\delta)_a$ codes when
$n\geq\alpha(r+\delta-1)$, $\frac{k}{3}<r<k$ and
$r\neq\frac{k}{2}$.

%\begin{exam}
%Suppose $n=28, k=11, r=5, \delta=4$
%\end{exam}

%\begin{exam}
%Suppose $n=29, k=9, r=5, \delta=4$
%\end{exam}

%\begin{exam}
%Suppose $n=45, k=18, r=5, \delta=4$
%\end{exam}

\vspace{0.2cm}\begin{center}
\begin{tabular}{|p{3.5mm}|p{3.5mm}|p{3.5mm}|p{3.5mm}|p{3.5mm}|p{3.5mm}
|p{3.5mm}|p{3.5mm}|p{3.5mm}|p{3.5mm}|p{3.5mm}|}
\hline $r\backslash k$ & \small11 & \small12 & \small13 & \small14 & \small15
& \small16 & \small17 & \small18 & \small19 & \small20 \\
\hline \small$2$ & \small{E$_M$} & \small{E$_M$} & \small{E$_M$} &
\small{E$_M$} & \small{E$_M$} & \small{E$_M$} & \small{E$_M$}
& \small{E$_M$} & \small{E$_M$} & \small{E$_M$} \\
\hline \small$3$ & \small N$_{11}$ & \small N$_{10}$ & \small
E$_{27}$ & \small E$_{27}$ & \small N$_{10}$ & \small N$_{11}$ &
\small N$_{11}$ & \small N$_{10}$ & \small N$_{11}$
& \small N$_{11}$ \\
\hline \small $4$ & \small E$_{27}$ & \small N$_{10}$ & \small
E$_{27}$ & \small E$_{27}$ & \small N$_{11}$ & \small N$_{10}$ &
\small E$_{27}$ & \small E$_{27}$ & \small N$_{11}$
& \small N$_{10}$ \\
\hline \small $5$ & \small E$_{16}$ & \small E$_{27}$ & \small
E$_{27}$ & \small E$_{27}$ & \small N$_{10}$ & \small E$_{27}$ &
\small E$_{27}$ & \small E$_{27}$ & \small N$_{12}$
& \small N$_{10}$ \\
\hline \small $6$ & \small{E$_M$} & \small{E$_M$} & \small{E$_M$}
& \small{E$_M$} & \small{E$_M$}
& \small{E$_M$} & \small{E$_M$} & \small{E$_M$} & \small{E$_M$} & \small{E$_M$} \\
\hline \small $7$ & \small E$_{26}$ & \small E$_{26}$ & \small
E$_{26}$ & \small N$_{10}$ & \small E$_{26}$ & \small E$_{26}$ &
\small E$_{26}$ & \small E$_{26}$
& \small E$_{26}$ & \small $\sim$ \\
\hline \small $8$ & \small{E$_M$} & \small{E$_M$} & \small{E$_M$}
& \small{E$_M$} & \small{E$_M$}
& \small{E$_M$} & \small{E$_M$} & \small{E$_M$} & \small{E$_M$} & \small{E$_M$} \\
\hline \small $9$ & \small E$_{16}$ & \small E$_{16}$ & \small
E$_{16}$ & \small E$_{26}$ & \small E$_{26}$ & \small E$_{26}$ &
\small E$_{26}$ & \small N$_{10}$ & \small E$_{16}$
& \small E$_{16}$ \\
\hline \small $10$ & \small $\sim$ & \small $\sim$ & \small $\sim$
& \small $\sim$ & \small $\sim$
& \small $\sim$ & \small $\sim$ & \small $\sim$ & \small $\sim$ & \small N$_{10}$ \\
\hline \small $11$ & \small{E$_M$} & \small{E$_M$} & \small{E$_M$}
& \small{E$_M$} & \small{E$_M$}
& \small{E$_M$} & \small{E$_M$} & \small{E$_M$} & \small{E$_M$} & \small{E$_M$} \\
\hline
\end{tabular}

\vspace{0.5cm}\footnotesize{Table 1. ~ Existence of optimal
$(r,\delta)_a$ codes for parameters $n=60,\delta=5, 2\leq r\leq
11$ and $11\leq k\leq 20$.}
\end{center}

\section{Conclusions}
We have investigated the structure properties and construction methods
of optimal $(r,\delta)_a$ linear codes, whose length and
dimension  are  $n$ and $k$ respectively.
A structure theorem for optimal $(r,\delta)_a$ code with $r|k$ is first
obtained. We next derived two sets of parameters where no
optimal $(r,\delta)_a$ linear codes could exist (over any field), as well as identified four sets of parameters where optimal $(r,\delta)_a$ linear codes
exist over any field of size $q\geq\binom{n}{k-1}$. Some of these existence conditions were reported in the literature before, but the minimum field size we derived is (considerably) smaller than those derived in the previous works.
Our results have considerably substantiated the results in terms of constructing  optimal $(r,\delta)_a$ codes, and there are now only two small holes (two subcases with specific parameters) where the existence results are unknown.
 Except for these two small subcases, for all the other cases, given each tuple of $(n,k,r,\delta)$, either an
optimal $(r,\delta)_a$ linear code does not exist or an optimal $(r,\delta)_a$
linear code can be constructed using a deterministic algorithm.

As an illustrative summary of our results, we also provide in
Table 1 an example of the existence of optimal $(r,\delta)_a$
linear codes for the parameters of $n=60$, $\delta=5, 2\leq r\leq
11$ and $11\leq k\leq 20$. In this table, E$_M$  means that
optimal $(r,\delta)_a$ linear codes can be constructed by the
method in \cite{Prakash12} or by our Theorem \ref{opt-ext-1} and
Algorithm 1 (which requires a substantially smaller field);
E$_{16}~($resp. E$_{26}$, E$_{27})$ means optimal $(r,\delta)_a$
linear codes can be constructed by Theorem \ref{opt-ext-2}
$($resp. Theorem \ref{opt-ext-3}, Theorem \ref{opt-ext-4}$)$;
N$_{10}~($resp. N$_{11})$ means optimal $(r,\delta)_a$ linear
codes do not exist according to Theorem \ref{non-exst} $($resp.
Theorem \ref{non-exst-1}$)$; and $\sim$ means we do not yet know
whether an optimal $(r,\delta)_a$ linear code exists or not.

\end{document}